\documentclass[11pt]{article} 

\usepackage[english]{babel}
\usepackage[T1]{fontenc} 
\usepackage[utf8]{inputenc}

\usepackage{color}
\usepackage{amsthm,amsmath,amssymb}

\usepackage{vmargin}
\setmarginsrb{3cm}{2cm}{3cm}{2cm}{0cm}{2cm}{0cm}{1cm}

\usepackage{rotating}
\usepackage{graphicx}
\usepackage{tikz}
\usetikzlibrary{arrows, automata, shapes, positioning, decorations}

\definecolor{bientotlafin}{RGB}{142, 162, 198}
\definecolor{green}{RGB}{31,160,85}
\definecolor{vert}{RGB}{0,255,127}

\theoremstyle{plain}
\newtheorem{theorem}{Theorem}
\newtheorem{lemma}[theorem]{Lemma}
\newtheorem{corollary}[theorem]{Corollary}
\newtheorem{proposition}[theorem]{Proposition}

\theoremstyle{definition}
\newtheorem{definition}[theorem]{Definition}
\newtheorem{example}[theorem]{Example}
\newtheorem{remark}[theorem]{Remark}

\newcommand{\N}{\mathbb N}
\newcommand{\rep}{\mathrm{rep}}
\newcommand{\val}{\mathrm{val}}
\newcommand{\card}{\mathrm{Card}}
\newcommand{\A}{\mathcal{A}}
\newcommand{\T}{\mathcal{T}}
\newcommand{\DIV}{\mathrm{DIV}}
\newcommand{\MOD}{\mathrm{MOD}}
\newcommand{\andrm}{\ {\rm and}\ }

\title{State complexity of the multiples of the Thue-Morse set}
\author{Émilie Charlier, Célia Cisternino, Adeline Massuir}

\begin{document}
\maketitle

\begin{abstract}
The Thue-Morse set is the set of those nonnegative integers whose binary expansions have an even number of $1$. We obtain an exact formula for the state complexity of the multiplication by a constant of the Thue-Morse set  $\mathcal{T}$ with respect with any base $b$ which is a power of $2$. Our proof is constructive and we are able to explicitly provide the minimal automaton of the language of all $2^p$-expansions of the set of integers $m\mathcal{T}$ for any positive integers $m$ and $p$.
\end{abstract}

\section{Introduction}

This paper is a contribution to the study of recognizable sets of integers. Many descriptions of such sets were given by various authors. Among them, we point~\cite{BHMV,Cobham1969,Eilenberg}.  A complete description of the minimal automaton recognizing $m\N$ in any given base $b$ was given in~\cite{Alexeev}. Structural properties of minimal automata recognizing $m\N$ are known in various non-standard numeration systems as well~\cite{CRRW}. A deep knowledge of the structures of such automata is important. For example, they can be fruitfully used to obtain efficient decision procedures of periodicity problems~\cite{BMMR,Marsault-Saka}. In the present work, we propose ourselves to initiate a study of the state complexity of the multiplication by a constant of recognizable subsets $X$ of $\N$. In doing so, we aim at generalizing the previous framework concerning the case $X=\N$ only.  Our study starts with the well-known Thue-Morse set $\T$ consisting of the natural numbers whose base $2$-expansions contain an even number of occurrences of the digit $1$. Our goal here is to provide a complete characterization of the minimal automata recognizing the sets $m\T$ for any multiple $m$ and any base $b$ which is a power of $2$.

\section{Basics}

In this text, we use the usual definitions and notation (alphabet, letter, word, language, free monoid, automaton, etc.) of formal language theory; for example, see \cite{Lothaire1997,Sakarovitch2009}.  

Nevertheless, let us give a few definitions and properties that will be central in this work. The empty word is denoted by $\varepsilon$. For a finite word $w$, $|w|$ designates its length and $|w|_a$ the number of occurrences of the letter $a$ in $w$.  A {\em regular language} is a language which is accepted by a finite automaton. For $L\subseteq A^*$ and $w\in A^*$, the {\em (left) quotient} of $L$ by $w$ is the language
\[
	w^{-1}L=\{u\in A^*\colon wu\in L\}.
\]
As is well known, a language $L$ over an alphabet $A$ is regular if and only if it has finitely many quotients, that is, the set of languages
\[
	\{w^{-1}L\colon w\in A^*\}
\]
is finite. The {\em state complexity} of a regular language is the number of its quotients: $\card(\{w^{-1}L\colon w\in A^*\})$. It corresponds to the number of states of its minimal automaton. The following characterization of minimal automata will be used several times in this work: a deterministic finite automaton (or DFA for short) is minimal if and only if it is complete, reduced and accessible. A DFA is said to be {\em complete} if the transition function is total (i.e.\ from every state start transitions labeled with all possible letters), {\em reduced} if languages accepted from distinct states are distinct and {\em accessible} if every state can be reached from the initial state. The language accepted from a state $q$ is denoted by $L_q$. Thus, the language accepted by a DFA is the language accepted from its initial state (we always consider automata having a single initial state).

In what follows we will need a notion that is somewhat stronger than that of reduced DFAs. We say that a DFA has {\em disjoint states} if the languages accepted from distinct states are disjoint: for distinct states $p$ and $q$, we have $L_p\cap L_q=\emptyset$. A state $q$ is said to be {\em coaccessible} if $L_q\ne \emptyset$ and, by extension, an automaton is {\em coaccessible} if all its states are coaccessible. Thus, any coaccessible DFA having disjoint states is reduced.

Now, let us give some background on numeration systems. Let $b\in\N_{\ge2}$. We define $A_b$ to be the alphabet $\{\tt{0},\ldots,\tt{b{-}1}\}$. Elements of $A_b$ are called {\em digits}. The number $b$ is called the {\em base} of the numeration. In what follows we will make no distinction between a digit ${\tt c}$ in $A_b$ and its {\em value} $c$ in $[\![0,b{-}1]\!]$. Otherwise stated, we identify the alphabet $A_b$ and the interval of integers $[\![0,b{-}1]\!]$. Note that here and throughout the text, we use the notation $[\![m,n]\!]$ to designate the interval of integers $\{m,m+1,\ldots,n\}$. The {\em $b$-expansion} of a positive integer $n$, which is denoted by $\rep_b(n)$, is the finite word $c_{\ell{-}1}\cdots c_0$ over $A_b$ defined by
\[
	n=\sum_{j=0}^{\ell{-}1} c_j b^j, \quad c_{\ell{-}1}\ne 0.
\]
The {\em  $b$-expansion} of $0$ is the empty word: $\rep_b(0)=\varepsilon$. Conversely, for a word $w=c_{\ell{-}1}\cdots c_0$ over $A_b$, we write $\val_b(w)=\sum_{j=0}^{\ell{-}1} c_j b^j$. Thus we have $\rep_b\colon\N\to A_b^*$ and $\val_b\colon A_b^*\to \N$. Clearly, the function $\val_b\circ \rep_b$ is the identity from $\N$ to $\N$. Moreover, for any $w\in A_b^*$, the words $\rep_b(\val_b(w))$ and $w$ only differ by the potential leading zeroes in $w$. Also note that for all subsets $X$ of $\N$, we have $\val_b^{-1}(X)=0^*\rep_b(X)$. A subset $X$ of $\N$ is said to be {\em $b$-recognizable} if the language $\rep_b (X)$ is regular. In what follows, we will always consider automata accepting $\val_b^{-1}(X)$ instead of $\rep_b(X)$. The {\em state complexity} of a $b$-recognizable subset $X$ of $\N$ {\em with respect to the base $b$} is the state complexity of the language $\val_b^{-1}(X)$. 

We will need to represent not only natural numbers, but also pairs of natural numbers. If $u=u_1\cdots u_n\in A^*$ and $v=v_1\cdots v_n\in B^*$ are words of the same length $n$, then we use the notation $(u,v)$ to designate the word $(u_1,v_1)\cdots (u_n,v_n)$ of length $n$ over the alphabet $A\times B$: 
\[
	(u,v)=(u_1,v_1)\cdots (u_n,v_n)\in (A\times B)^*.
\]
For $(m,n)\in\N^2$, we write
\[
	\rep_b(m,n)=(0^{\ell-|\rep_b(m)|}\rep_b(m),0^{\ell-|\rep_b(n)|}\rep_b(n))
\] 
where $\ell=\max\{|\rep_b(m)|,|\rep_b(n)|\}$. Otherwise stated, we add leading zeroes to the shortest expansion (if any) in order to obtain two words of the same length. Finally, for a subset $X$ of $\N^2$, we write
\[
	\val_{b}^{-1}(X)=(0,0)^*\rep_b(X).
\]

\section{Method}\label{sec:objectifmethode}

The Thue-Morse set, which we denote by $\T$, is the set of all natural numbers whose base-$2$ expansions contain an even number of occurrences of the digit $1$:
\[
	\T= \{n\in\N\colon|\rep_2 (n)|_1\in 2\N\}.
\]
The Thue-Morse set $\T$ is $2$-recognizable since the language $\val_2^{-1}(\T)$ is accepted by the automaton depicted in Figure~\ref{fig:aut-TM-2}.
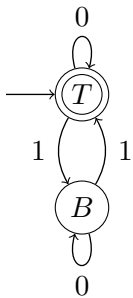
\begin{figure}[htb]
\centering
\begin{tikzpicture}
\tikzstyle{every node}=[shape=circle, fill=none, draw=black,
minimum size=20pt, inner sep=2pt]
\node(1) at (0,0) {$T$};
\node(2) at (0,-1.5) {$B$};
\tikzstyle{every node}=[shape=circle, fill=none, draw=black,
minimum size=15pt, inner sep=2pt]
\node(2f) at (0,0) {};
\tikzstyle{every path}=[color=black, line width=0.5 pt]
\tikzstyle{every node}=[shape=circle, minimum size=5pt, inner sep=2pt]
\draw [->] (-1,0) to node {} (1); 
\draw [->] (1) to [loop above] node [above] {$0$} (1);
\draw [->] (2) to [loop below] node [below] {$0$} (2);
\draw [->] (1) to [bend right=30] node [left] {$1$} (2);
\draw [->] (2) to [bend right=30] node [right] {$1$} (1);
\end{tikzpicture}
\caption{The Thue-Morse set is $2$-recognizable.}
\label{fig:aut-TM-2}
\end{figure}
More precisely, the Thue-Morse set $\T$ is $2^p$-recognizable for all $p\in\N_{\ge1}$ and is not $b$-recognizable for any other base $b$. This is a consequence of the famous theorem of Cobham.

Two positive integers are said to be {\em multiplicatively independent} if their only common integer power is $1$.

\begin{theorem}[\cite{Cobham1969}] \ 
\begin{itemize}
\item Let $b,b'$ be two multiplicatively independent bases. Then a subset of $\N$ is both $b$-recognizable and $b'$-recognizable if and only if it is a finite union of arithmetic progressions.
\item Let $b,b'$ be two multiplicatively dependent bases. Then a subset of $\N$ is $b$-recognizable if and only if it is $b'$-recognizable.
\end{itemize}
\end{theorem}

In the case of the Thue-Morse set, it is easily seen that, for each $p\in\N_{\ge1}$, the language $\val_{2^p}^{-1}(\T)$ is accepted by the DFA $(\{T,B\},T,T,A_{2^p},\delta)$ where for all $X\in\{T,B\}$ and all $a\in A_{2^p}$,  
\[
	\delta(X,a)=\begin{cases}
					X & \text{if } a\in \T \\
					\overline{X} & \text{else}
				\end{cases}
\]
where $\overline{T}=B$ and $\overline{B}=T$. For example this automaton is depicted in Figure~\ref{fig:aut-TM-4} for $p=2$.
\begin{figure}[htb]
\centering
\begin{tikzpicture}
\tikzstyle{every node}=[shape=circle, fill=none, draw=black,
minimum size=20pt, inner sep=2pt]
\node(1) at (0,0) {$T$};
\node(2) at (0,-1.5) {$B$};
\tikzstyle{every node}=[shape=circle, fill=none, draw=black,
minimum size=15pt, inner sep=2pt]
\node(2f) at (0,0) {};
\tikzstyle{every path}=[color=black, line width=0.5 pt]
\tikzstyle{every node}=[shape=circle, minimum size=5pt, inner sep=2pt]
\draw [->] (-1,0) to node {} (1); 
\draw [->] (1) to [loop above] node [above] {$0,3$} (1);
\draw [->] (2) to [loop below] node [below] {$0,3$} (2);
\draw [->] (1) to [bend right=30] node [left] {$1,2$} (2);
\draw [->] (2) to [bend right=30] node [right] {$1,2$} (1);
\end{tikzpicture}
\caption{The Thue-Morse set is $4$-recognizable.}
\label{fig:aut-TM-4}
\end{figure}
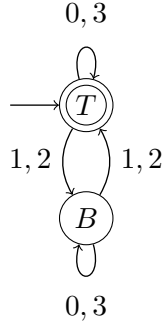

In order to avoid a systematic case separation, we introduce the following notation: for $X\in\{T,B\}$ and $n \in \N$, we define
\[
	X_n = \begin{cases}
			X & \text{if } n \in \T \\
			\overline{X} & \text{else.}  
			\end{cases}
\]
With this notation, we can simply rewrite the definition of the transition function $\delta$ as $\delta(X,a)=X_a$.

The following proposition is well known; for example see~\cite{BHMV}. 

\begin{proposition}
Let $b\in\N_{\ge 2}$ and $m\in\N$. If $X$ is $b$-recognizable, then so is $mX$. Otherwise stated, multiplication by a constant preserves $b$-recogniza\-bility.
\end{proposition}

In particular, for any $m\in\N$ and $p\in\N_{\ge 1}$, the set $m\T$ is $2^p$-recognizable. The aim of this work is to show the following result.

\begin{theorem} \label{thm:main}
Let $m$ and $p$ be positive integers. Then the state complexity of $m\T$ with respect to the base $2^p$ is equal to 
\[
	2k+\left\lceil \frac zp\right\rceil
\]
if $m=k2^z$ with $k$ odd.
\end{theorem}

Our proof of Theorem~\ref{thm:main} is constructive. In order to describe the minimal DFA of $\val_{2^p}^{-1}(m\T)$, we will successively construct  several automata. First, we build a DFA $\A_{\T,2^p}$ accepting the language 
\[
	\val_{2^p}^{-1}(\T\times \N).
\]
Then we build a DFA $\A_{m,b}$ accepting the language
\[
	\val_b^{-1}\big(\{(n,mn)\colon n\in \N\}\big).
\] 
Note that we do the latter step for any integer base $b$ and not only for powers of $2$.
Next, we consider the product automaton $\A_{m,2^p}\times\A_{\T,2^p}$. This DFA accepts the language 
\[
	\val_{2^p}^{-1}\big(\{(t,mt)\colon t \in \T\}\big).
\]
Finally, a finite automaton $\Pi(\A_{m,2^p}\times\A_{\T,2^p})$ accepting $\val_{2^p}^{-1}(m\T)$ is obtained by projecting the label of each transition in $\A_{m,2^p}\times\A_{\T,2^p}$ onto its second component.
At each step of our construction, we check that the automaton under consideration is minimal (and hence deterministic) and the ultimate step precisely consists in a minimization procedure.

From now on, we fix some positive integers $m$ and $p$. We also let $z$ and $k$ be the unique integers such that $m=k2^z$ with $k$ odd.

\section{The automaton $\A_{\T,2^p}$}

In this section, we build and study a DFA  $\A_{\T,2^p}$ accepting the language $\val_{2^p}^{-1}(\T\times \N)$. This DFA is a modified version of the automaton accepting $\val_{2^p}^{-1}(\T)$ defined in the previous section. Namely, we replace each transition labeled by $a\in A_{2^p}$ by $2^p$ copies of itself labeled by $(a,b)$, for each $b\in A_{2^p}$. Formally, 
\[
	\A_{\T,2^p}=(\{T,B\},T,T,A_{2^p}\times A_{2^p},\delta_{\T,2^p})
\]
where, for all $X\in\{T,B\}$ and all $a,b\in A_{2^p}$, we have $\delta_{\T,2^p}(X,(a,b))=X_a$. For example, the automata $\A_{\T,2}$ and $\A_{\T,4}$ are depicted in Figure~\ref{fig:aut-TM-2-4-couples}.
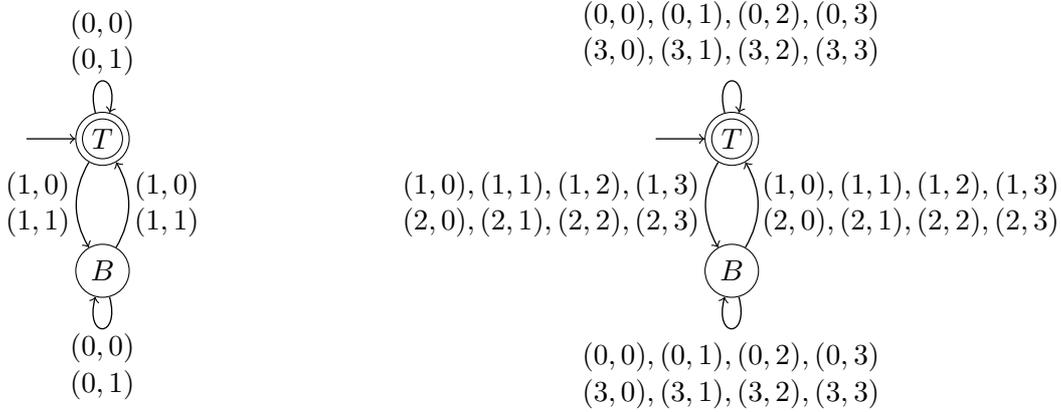
\begin{figure}[htb]
\begin{minipage}[c]{0.35\linewidth}
\begin{tikzpicture}
\tikzstyle{every node}=[shape=circle, fill=none, draw=black,
minimum size=20pt, inner sep=2pt]
\node(1) at (0,0) {$T$};
\node(2) at (0,-1.75) {$B$};
\tikzstyle{every node}=[shape=circle, fill=none, draw=black,
minimum size=15pt, inner sep=2pt]
\node(2f) at (0,0) {};
\tikzstyle{every path}=[color=black, line width=0.5 pt]
\tikzstyle{every node}=[shape=circle, minimum size=5pt, inner sep=2pt]
\draw [->] (-1,0) to node {} (1); 
\draw [->] (1) to [loop above] node [above=-0.4cm] {\begin{tabular}{c}
$(0,0)$ \\
$(0,1)$ \\
\end{tabular}} (1);
\draw [->] (2) to [loop below] node [below=-0.4cm] {\begin{tabular}{c}
$(0,0)$ \\
$(0,1)$ \\
\end{tabular}} (2);
\draw [->] (1) to [bend right=30] node [left=-0.4cm] {\begin{tabular}{c}
$(1,0)$ \\
$(1,1)$ \\
\end{tabular}} (2);
\draw [->] (2) to [bend right=30] node [right=-0.4cm] {\begin{tabular}{c}
$(1,0)$ \\
$(1,1)$ \\
\end{tabular}} (1);
\end{tikzpicture}
\end{minipage}
\begin{minipage}[c]{0.65\linewidth}
\begin{tikzpicture}
\tikzstyle{every node}=[shape=circle, fill=none, draw=black,
minimum size=20pt, inner sep=2pt]
\node(1) at (0,0) {$T$};
\node(2) at (0,-1.75) {$B$};
\tikzstyle{every node}=[shape=circle, fill=none, draw=black,
minimum size=15pt, inner sep=2pt]
\node(2f) at (0,0) {};
\tikzstyle{every path}=[color=black, line width=0.5 pt]
\tikzstyle{every node}=[shape=circle, minimum size=5pt, inner sep=2pt]
\draw [->] (-1,0) to node {} (1); 
\draw [->] (1) to [loop above] node [above=-1.7cm] {\begin{tabular}{c}
$(0,0),(0,1),(0,2),(0,3)$ \\
$(3,0),(3,1),(3,2),(3,3)$ \\
\end{tabular}} (1);
\draw [->] (2) to [loop below] node [below=-1.7cm] {\begin{tabular}{c}
$(0,0),(0,1),(0,2),(0,3)$ \\
$(3,0),(3,1),(3,2),(3,3)$ \\
\end{tabular}} (2);
\draw [->] (1) to [bend right=30] node [left=-0.3cm] {\begin{tabular}{c}
$(1,0),(1,1),(1,2),(1,3)$ \\
$(2,0),(2,1),(2,2),(2,3)$ \\
\end{tabular}} (2);
\draw [->] (2) to [bend right=30] node [right=-0.3cm] {\begin{tabular}{c}
$(1,0),(1,1),(1,2),(1,3)$ \\
$(2,0),(2,1),(2,2),(2,3)$ \\
\end{tabular}} (1);
\end{tikzpicture}
\end{minipage}
\caption{The automata $\A_{\T,2}$ (left) and $\A_{\T,4}$ (right).}
\label{fig:aut-TM-2-4-couples}
\end{figure}

Now we prove some properties of the automaton $\A_{\T,2^p}$ that will be useful for our concerns. 

\begin{lemma}
The automaton $\A_{\T,2^p}$ is complete, accessible, coaccessible and has disjoint states. In particular, it is the minimal automaton of $\val_{2^p}^{-1}(\T\times \N)$.
\end{lemma}

\begin{proof}
These properties are all straightforward verifications.
\end{proof}

\begin{lemma}\label{lemlem:transitionsTM}
Let $u,v \in A_{2^p}^*$. Then $\val_{2^p}(uv)\in\T$ if and only if, either $\val_{2^p}(u)\in\T$ and  $\val_{2^p}(v)\in\T$, or $\val_{2^p}(u)\notin\T$ and  $\val_{2^p}(v)\notin\T$.
\end{lemma}

\begin{proof}
Let $\tau\colon A_{2^p}^*\to A_{2^p}^*$ be the $p$-uniform morphism defined by $\tau(a)=0^{p-|\rep_2(a)|}\rep_2(a)$ for each $a\in A_{2^p}$. Then, for all $w\in A_{2^p}^*$, we have $\val_{2^p}(w)=\val_2(\tau(w))$. Therefore, $\val_{2^p}(w)\in\T$ if and only if $|\tau(w)|_1\in 2\N$.
Since $\tau$ is a morphism, we have $|\tau(uv)|_1=|\tau(u)|_1+|\tau(v)|_1$. Hence $|\tau(uv)|_1$ is even if and only if $|\tau(u)|_1$ and $|\tau(v)|_1$ are both even or both odd. 
\end{proof}

\begin{lemma}\label{lem:transitionsTM}
For all $X \in \{T,B\}$ and $(u,v)\in (A_{2^p}\times A_{2^p})^*$, we have 
\[
	\delta_{\T,2^p}(X,(u,v))=X_{\val_{2^p}(u)}.
\]
\end{lemma}

\begin{proof}
We do the proof by induction on $|(u,v)|$. The case $|(u,v)|=0$ is trivial. The case $|(u,v)|=1$ holds by definition of $\A_{\T,2^p}$. Now let $X\in \{T,B \}$ and let $(ua,vb)\in (A_{2^p}\times A_{2^p})^*$ with $a,b\in A_{2^p}$. We suppose that the result is satisfied for $(u,v)$ and we show that it is also true for $(ua,vb)$. Let $Y=\delta_{\T,2^p}(X,(u,v))$. By induction hypothesis, we have $Y =X_{\val_{2^p}(u)}$. Thus we obtain
\[
	\delta_{\T,2^p}(X,(ua,vb)) 
	=\delta_{\T,2^p}(Y,(a,b))
	=Y_a
	=(X_{\val_{2^p}(u)})_a
	=X_{\val_{2^p}(ua)}.
\]
where we have used Lemma~\ref{lemlem:transitionsTM} for the last step.
\end{proof}

\section{The automaton $\A_{m,b}$}
In this section, we consider an arbitrary integer base $b$. Let
\[
	\A_{m,b}=([\![0,m{-}1]\!],0,0,A_b\times A_b,\delta_{m,b})
\]
where the (partial) transition function $\delta_{m,b}$ is defined as follows: for $i,j\in[\![0,m{-}1]\!]$ and $d,e\in A_b$, we set
\[
	\delta_{m,b}(i,(d,e))=j \quad \iff \quad bi + e = md + j.
\]
The DFA $\A_{m,b}$ accepts the language $\val_{b}^{-1}\big(\{(n,mn)\colon n\in \N\}\big)$. We refer the interested reader to \cite{Waxweiler2009}. For example, the automaton $\A_{6,4}$ is depicted in Figure~\ref{fig:A-6,4}. 

\begin{figure}[htb]
\centering
\begin{tikzpicture}[scale=0.75]
\tikzstyle{every node}=[shape=circle, fill=none, draw=black,
minimum size=30pt, inner sep=2pt]
\node(0) at (0,6.5) {$0$};
\node(1) at (3.5,6.5) {$1$};
\node(2) at (7,6.5) {$2$};
\node(3) at (10.5,6.5) {$3$};
\node(4) at (14,6.5) {$4$};
\node(5) at (17.5,6.5) {$5$};

\tikzstyle{every node}=[shape=circle, fill=none, draw=black,
minimum size=25pt, inner sep=2pt]
\node at (0,6.5) {};

\tikzstyle{etiquettedebut}=[very near start,rectangle,fill=black!20,scale=0.9]
\tikzstyle{etiquettemilieu}=[midway,rectangle,fill=black!20,scale=0.9]
\tikzstyle{every path}=[color=black, line width=0.5 pt,scale=0.9]
\tikzstyle{every node}=[shape=circle, minimum size=5pt, inner sep=2pt,scale=0.9]

\draw [->] (-1.5,7.22) to node {} (0); 
\draw [->] (0) to [loop above] node [left,rectangle,fill=black!20,scale=0.9] {$(0,0)$} (0);
\draw [->] (1) to [loop below] node [left,rectangle,fill=black!20,scale=0.9] {$(1,3)$} (1);
\draw [->] (2) to [loop left] node [left,pos=0.2,rectangle,fill=black!20,scale=0.9] {$(1,0)$} (2);
\draw [->] (3) to [loop below] node [right,pos=0.3,rectangle,fill=black!20,scale=0.9] {$(2,3)$} (3);
\draw [->] (4) to [loop above] node [left,rectangle,fill=black!20,scale=0.9] {$(2,0)$} (4);
\draw [->] (5) to [loop above] node [left,rectangle,fill=black!20,scale=0.9] {$(3,3)$} (5);
\draw [->] (0) to [bend left=17] node [sloped,etiquettemilieu] {$(0,1)$} (1);
\draw [->] (0) to [bend left=30] node [sloped,etiquettemilieu] {$(0,2)$} (2);
\draw [->] (0) to [bend left=45] node [sloped,etiquettemilieu] {$(0,3)$} (3);
\draw [->] (1) to [bend left=40] node [sloped,etiquettemilieu] {$(0,0)$} (4);
\draw [->] (1) to [bend left=50] node [sloped,etiquettemilieu] {$(0,1)$} (5);
\draw [->] (1) to [bend left=17] node [sloped,etiquettemilieu] {$(1,2)$} (0);
\draw [->] (2) to [bend left=17] node [sloped,etiquettemilieu] {$(1,1)$} (3);
\draw [->] (2) to [bend left=30] node [sloped,etiquettemilieu] {$(1,2)$} (4);
\draw [->] (2) to [bend left=40] node [sloped,etiquettemilieu] {$(1,3)$} (5);
\draw [->] (3) to [bend left=40] node [sloped,etiquettemilieu] {$(2,0)$} (0);
\draw [->] (3) to [bend left=30] node [sloped,etiquettemilieu] {$(2,1)$} (1);
\draw [->] (3) to [bend left=17] node [sloped,etiquettemilieu] {$(2,2)$} (2);
\draw [->] (4) to [bend left=17] node [sloped,etiquettemilieu] {$(2,1)$} (5);
\draw [->] (4) to [bend left=50] node [sloped,etiquettemilieu] {$(3,2)$} (0);
\draw [->] (4) to [bend left=40] node [sloped,etiquettemilieu] {$(3,3)$} (1);
\draw [->] (5) to [bend left=45] node [sloped,etiquettemilieu] {$(3,0)$} (2);
\draw [->] (5) to [bend left=35] node [sloped,etiquettemilieu] {$(3,1)$} (3);
\draw [->] (5) to [bend left=17] node [sloped,etiquettemilieu] {$(3,2)$} (4);
\end{tikzpicture}
\caption{The automaton $\A_{6,4}$ accepts the language $\val_4^{-1}\big(\{(n,6n)\colon n\in \N\}\big)$.}
\label{fig:A-6,4}
\end{figure}
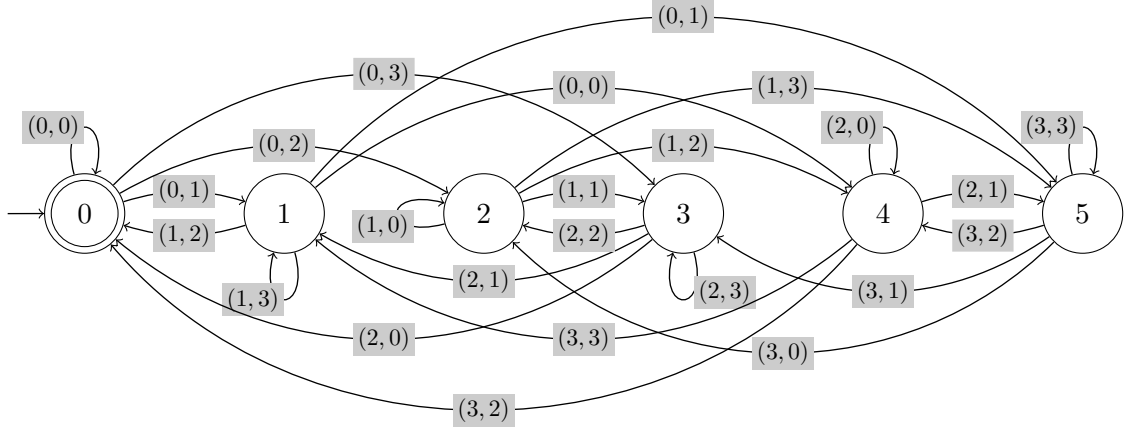

Note that the automaton $\A_{m,b}$ is not complete (see Remark~\ref{rem:unicitelettrepremcomp}). Also note that there is always a loop labeled by $(0,0)$ on the initial state $0$.

\begin{remark}
\label{rem:unicitelettrepremcomp}
For each $i \in [\![0,m{-}1]\!]$ and $e \in A_b$, there exist unique $d \in A_b$ and $j \in [\![0,m{-}1]\!]$ such that
$\delta_{m,b}(i,(d,e))=j$. Indeed, $d$ and $j$ are unique since they are the quotient and remainder of the Euclidean division of $bi+e$ by $m$. We still have to check that $d < b$. We have
\[
	bi + e = md + j \iff d = \frac{bi+e-j}{m}.
\]
Since $i\le m{-}1$, $j \ge 0$ and $e < b$, we have
\[
   	\frac{bi+e-j}{m}\le \frac{b(m{-}1)+e}{m}=b-\frac{b-e}{m}<b.
\]
\end{remark}

\begin{lemma} \label{lem:transitionsAmb}
For $i,j\in [\![0,m{-}1]\!]$ and $(u,v) \in (A_b\times A_b)^*$, we have
\[
	\delta_{m,b}(i,(u,v))=j \iff b^{|(u,v)|}\, i + \val_b(v) = m\, \val_b(u) + j.
\]
\end{lemma}

\begin{proof}
We do the proof by induction on $n=|(u,v)|$. If $n$ is equal to $0$ or $1$, the result is clear. Now let $i,j\in [\![0,m{-}1]\!]$ and let $(du,ev)\in (A_b\times A_b)^*$ with $d,e\in A_b$ and $|(u,v)|=n$. We suppose that the result is satisfied for $(u,v)$ and we show that it is also true for $(du,ev)$. We use the notation $\DIV(x,y)$ and $\MOD(x,y)$ to designate the quotient and the remainder of the Euclidean division of $x$ by $y$ (thus, we have $\DIV(x,y)=\big\lfloor \frac xy \big\rfloor$). 
By definition of the transition function, we have
\[
	\delta_{m,b}(i,(du,ev))=j
	 \iff  d=\DIV(bi+e,m) \ \andrm \ \delta_{m,b}(\MOD(bi+e,m),(u,v))=j.
\]
By using the induction hypothesis, we have
\begin{align*}
	\delta_{m,b}(bi+e-md),(u,v))=j
	& \iff b^n\, (bi+e-md) + \val_b(v)= m\, \val_b(u)+j 		\\
	& \iff b^{n+1}\, i + \val_b(ev)= m\, \val_b(du)+j.
\end{align*}
To be able to conclude the proof, we still have to show that 
\begin{equation}
\label{eq:val}
	b^{n+1}\, i + \val_b(ev)= m\, \val_b(du)+j 
\end{equation}
implies 
\[
	d=\DIV(bi+e,m).
\]
Thus, suppose that \eqref{eq:val} is true. Then
\[
	b^{n+1}\, i+ b^n e + \val_b(v)= m(b^nd+ \val_b(u))+j.
\]
Since $\val_b(u)$ and $\val_b(v)$ are less than $b^n$, $d\ge 0$, $j<m$ and $b^nd+ \val_b(u)\ge 0$, 
we obtain
\begin{align*}
	d 	&= \DIV(b^nd+ \val_b(u),b^n) \\
		&= \DIV(\DIV(b^{n+1}\, i+ b^n e + \val_b(v),m),b^n) \\
		&= \DIV(\DIV(b^{n+1}\, i+ b^n e + \val_b(v),b^n),m) \\
		&= \DIV(b\, i+ e,m)
\end{align*}
as desired.
\end{proof}

\begin{remark} \label{rem:unicitepremcomp}
It is easily checked that Remark~\ref{rem:unicitelettrepremcomp} extends from letters to words: for each $i\in[\![0,m{-}1]\!]$ and $v\in A_b^*$, there exist unique $u\in A_b^*$ and $j\in[\![0,m{-}1]\!]$ such that $\delta_{m,b}(i,(u,v))=j$. In particular, the word $u$ must have the same length as the word $v$, and hence $\val_b(u)<b^{|v|}$.
\end{remark}

Let us describe a few properties of the automaton $\A_{m,b}$.

\begin{proposition} \label{prop:Ambproperties}
The automaton $\A_{m,b}$ is accessible, coaccessible and has disjoint states. 
\end{proposition}

\begin{proof} 
For each $i \in[\![0,m{-}1]\!]$, we have $\delta_{m,b}(0, \rep_b(0,i))=i$ from Lemma~\ref{lem:transitionsAmb}. Therefore $\A_{m,b}$ is accessible. It is a little trickier to find a word $(u,v)$ that leads from $i$ to $0$. The reason is that there is a length constraint to respect: we must find words $u,v\in A_b^*$ of the same length $n$ such that  $b^n i + \val_b(v) = m\, \val_b(u)$. Equivalently, we have to find $n\in\N$ and $d,e\in[\![0,b^n-1]\!]$ such that $b^n i+e=md$.

We claim that for all $n\in \N$ and $i \in[\![0,m{-}1]\!]$, there exists such $d$ and $e$ if and only if the following two inequalities hold
\begin{equation}\label{eqn:paslesmotslespluscourts}
	\left\lceil \frac{b^ni}{m} \right\rceil-\frac{b^n}{m} < 
\frac{b^ni}{m} \le b^n-1.
\end{equation}
First, suppose that $d,e\in[\![0,b^n-1]\!]$ are such that $b^n i+e=md$. Then $\frac{b^ni}{m}=d-\frac{e}{m}\le d\le b^n-1$. Moreover $\frac{b^ni}{m}\le d=\frac{b^n i+e}{m}<\frac{b^n (i+1)}{m}$. Since $d$ is an integer, we get that $\lceil \frac{b^ni}{m}\rceil<\frac{b^n(i+1)}{m}$. Conversely, suppose that the two inequalities~\eqref{eqn:paslesmotslespluscourts} hold. Let $d=\lceil\frac{b^ni}{m}\rceil$ and $e=md-b^ni$. It suffices to show that $d,e\in[\![0,b^n-1]\!]$. Clearly $d,e\in\N$. From the inequality on the right, we get $d\le b^n-1$ and from that on the left, we get $e=md-b^ni<b^n(i+1)-b^ni=b^n$. This proves the claim.
	
For a given $i \in[\![0,m{-}1]\!]$, the inequalities~\eqref{eqn:paslesmotslespluscourts} may not be satisfied for small $n$ but it is easily checked that they are both satisfied for all $n$ large enough. Therefore, the claim implies that $\A_{m,b}$ is coaccessible. 

Finally, let $i,j \in [\![0,m{-}1]\!]$ and let $(u,v) \in L_i\cap L_j$. By Lemma~\ref{lem:transitionsAmb}, we have
\[
	b^{|(u,v)|} i + \val_b(v) = m\, \val_b(u)\quad  \andrm \quad 
	 b^{|(u,v)|} j + \val_b(v) = m\,\val_b(u),
\]
which implies that $i=j$. We have thus obtained that $i\ne j\implies L_i\cap L_j=\emptyset$, i.e.\ that $\A_{m,b}$ has disjoint states.
\end{proof}

In a reduced DFA, there can be at most one non coaccessible state. Thus, we deduce from Proposition~\ref{prop:Ambproperties} that $\A_{m,b}$ is indeed the {\em trim minimal} automaton of the language $\val_b^{-1}\big(\{(n,mn)\colon n\in \N\}\big)$, that is the automaton obtained by removing the only non coaccessible state from its minimal automaton.

\section{The projected automaton $\Pi(\A_{m,b})$}

In this section, we study the automaton obtained by projecting the label of each transition of $\A_{m,b}$ onto its second component. We denote by $\Pi(\A_{m,b})$ the automaton obtained thanks to this projection. For example, the automaton $\Pi(\A_{6,4})$ is depicted in Figure~\ref{projA-6,4}.

\begin{figure}[htb]
\centering
\begin{tikzpicture}[scale=0.75]
\tikzstyle{every node}=[shape=circle, fill=none, draw=black,
minimum size=30pt, inner sep=2pt]
\node(0) at (0,6.5) {$0$};
\node(1) at (3.5,6.5) {$1$};
\node(2) at (7,6.5) {$2$};
\node(3) at (10.5,6.5) {$3$};
\node(4) at (14,6.5) {$4$};
\node(5) at (17.5,6.5) {$5$};

\tikzstyle{every node}=[shape=circle, fill=none, draw=black,
minimum size=25pt, inner sep=2pt]
\node at (0,6.5) {};

\tikzstyle{etiquettedebut}=[very near start,rectangle,fill=black!20]
\tikzstyle{etiquettemilieu}=[midway,rectangle,fill=black!20]
\tikzstyle{every path}=[color=black, line width=0.5 pt]
\tikzstyle{every node}=[shape=circle, minimum size=5pt, inner sep=2pt]

\draw [->] (-1.5,6.5) to node {} (0); 
\draw [->] (0) to [loop above] node [left,rectangle,fill=black!20] {$0$} (0);
\draw [->] (1) to [loop below] node [left,rectangle,fill=black!20] {$3$} (1);
\draw [->] (2) to [loop left] node [left,pos=0.2,rectangle,fill=black!20] {$0$} (2);
\draw [->] (3) to [loop below] node [right,pos=0.3,rectangle,fill=black!20] {$3$} (3);
\draw [->] (4) to [loop above] node [left,rectangle,fill=black!20] {$0$} (4);
\draw [->] (5) to [loop above] node [left,rectangle,fill=black!20] {$3$} (5);
\draw [->] (0) to [bend left=17] node [sloped,etiquettemilieu] {$1$} (1);
\draw [->] (0) to [bend left=30] node [sloped,etiquettemilieu] {$2$} (2);
\draw [->] (0) to [bend left=45] node [sloped,etiquettemilieu] {$3$} (3);
\draw [->] (1) to [bend left=40] node [sloped,etiquettemilieu] {$0$} (4);
\draw [->] (1) to [bend left=45] node [sloped,etiquettemilieu] {$1$} (5);
\draw [->] (1) to [bend left=17] node [sloped,etiquettemilieu] {$2$} (0);
\draw [->] (2) to [bend left=17] node [sloped,etiquettemilieu] {$1$} (3);
\draw [->] (2) to [bend left=30] node [sloped,etiquettemilieu] {$2$} (4);
\draw [->] (2) to [bend left=40] node [sloped,etiquettemilieu] {$3$} (5);
\draw [->] (3) to [bend left=40] node [sloped,etiquettemilieu] {$0$} (0);
\draw [->] (3) to [bend left=35] node [sloped,etiquettemilieu] {$1$} (1);
\draw [->] (3) to [bend left=17] node [sloped,etiquettemilieu] {$2$} (2);
\draw [->] (4) to [bend left=17] node [sloped,etiquettemilieu] {$1$} (5);
\draw [->] (4) to [bend left=50] node [sloped,etiquettemilieu] {$2$} (0);
\draw [->] (4) to [bend left=40] node [sloped,etiquettemilieu] {$3$} (1);
\draw [->] (5) to [bend left=40] node [sloped,etiquettemilieu] {$0$} (2);
\draw [->] (5) to [bend left=35] node [sloped,etiquettemilieu] {$1$} (3);
\draw [->] (5) to [bend left=17] node [sloped,etiquettemilieu] {$2$} (4);
\end{tikzpicture}
\caption{The projected automaton $\Pi(\A_{6,4})$.}
\label{projA-6,4}
\end{figure}
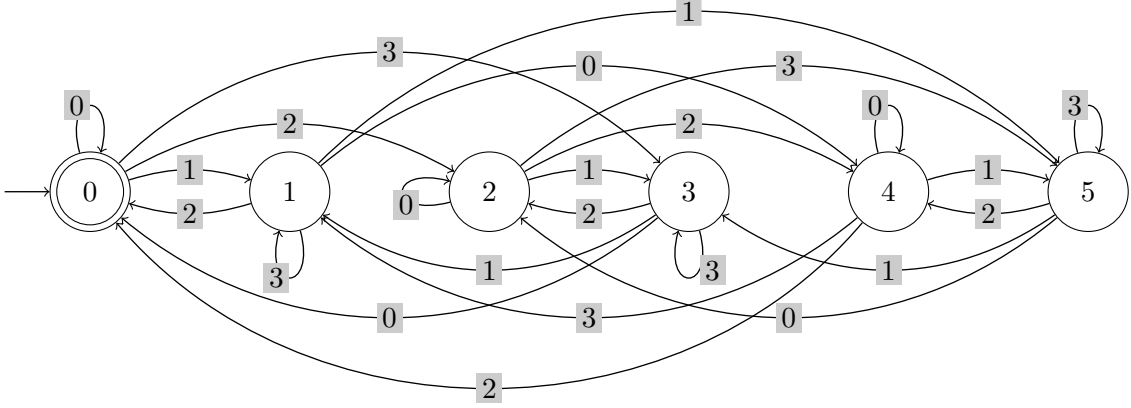

\begin{remark}
\label{rem:common}
The automaton $\Pi(\A_{m,b})$ corresponds to the automaton that is commonly built for accepting the language $\val_b^{-1}(m\N)$. For each $i,j\in[\![0,m{-}1]\!]$, there is a transition labeled by $e\in A_b$ from the state $i$ to the state $j$ if and only if $j=bi+e\bmod m$.
\end{remark}

\begin{corollary}
The automaton $\Pi(\A_{m,b})$ is complete, accessible and coaccessible.
\end{corollary}

\begin{proof}
The accessibility and coaccessibility of the automaton $\Pi(\A_{m,b})$ are straightforward consequences of Proposition~\ref{prop:Ambproperties}. The fact that it is complete comes from Remark~\ref{rem:common}: for every state $i\in[\![0,m{-}1]\!]$ and every digit $e\in A_b$, there is a transition labeled by $e$ from $i$ to the state $bi+e\bmod m$.
\end{proof}

The automaton $\Pi(\A_{m,b})$ is not minimal: it is minimal if and only if $m$ and $b$ are coprime; see for example \cite{Alexeev}. In fact, whenever $m$ and $b$ are coprime, we have a stronger property than minimality as shown in the following proposition. This result will be useful in our future considerations.

\begin{proposition}\label{prop:projAmbdisj}
If $m$ and $b$ are coprime, then the automaton $\Pi(\A_{m,b})$ has disjoint states, and hence it is the minimal automaton of $\val_b^{-1}(m\N)$.
\end{proposition}

\begin{proof}
Let $i,j\in[\![0,m{-}1]\!]$ and let $v \in A_b^*$ be a word accepted from both $i$ and $j$ in $\Pi(\A_{m,b})$. By Remark~\ref{rem:unicitepremcomp}, there exist unique words $u$ and $u'$ of the same length as $v$ such that $(u,v)$ and $(u',v)$ are accepted from $i$ and $j$ in $\A_{m,b}$ respectively. By Lemma~\ref{lem:transitionsAmb}, it is equivalent to say that
\[
	b^{|v|} i + \val_b (v) = m\, \val_b(u) 
	\quad \andrm  \quad 
	b^{|v|} j + \val_b (v) = m\, \val_b (u').
\]
Thus, we have
\begin{equation}\label{eqn:projAmbPropforte}
	b^{|v|} i - m\, \val_b (u) = b^{|v|} j -m\, \val_b( u').
\end{equation}
Therefore $m\, \val_b (u) \equiv m\, \val_b(u') \, \pmod{b^{|v|}}$. By using the hypothesis of coprimality of $m$ and $b$, we obtain that $\val_b (u) \equiv \val_b(u') \, \pmod{b^{|v|}}$. Since $\val_b(u)$ and $\val_b (u')$ are both less than $b^{|v|}$, we obtain the equality $\val_b (u) =\val_b (u')$. Finally, we get from \eqref{eqn:projAmbPropforte} that $i=j$, which proves that $\Pi(\A_{m,b})$ has disjoint states.
\end{proof}


To end this section, we prove some useful properties of the automaton $\Pi(\A_{m,b})$ under the more restrictive hypotheses of this work: $b=2^p$ and $m=k2^z$ with $k$ odd. 

\begin{lemma}
\label{lem:k>1}
If $k>1$ and $n=|\rep_{2^p}\big((k-1)2^z\big)|$, then $pn\ge z$.
\end{lemma}

\begin{proof}
Then
\[
	n=\left\lfloor \log_{2^p}\big((k-1)2^z)\right\rfloor +1 
	=\left\lfloor \log_{2^p}(k-1)+\frac zp\right\rfloor +1
	\ge \left\lfloor \frac zp\right\rfloor+1
	\ge \left\lceil \frac zp\right\rceil.
\]
Thus $pn\geq p\big\lceil \frac zp\big\rceil \geq z$.
\end{proof}

For $k>1$ and $n=|\rep_{2^p}\big((k{-}1)2^z\big)|$, we let $\sigma$ be the permutation of the integers in $[\![0,k{-}1]\!]$ defined by $\sigma(j)=-j2^{pn-z}\bmod k$. Further, we define 
\[
	w_j=0^{n-|\rep_{2^p}(\sigma(j)2^z)|}\rep_{2^p}(\sigma(j)2^z)
\] 
for each $j\in [\![0,k{-}1]\!]$. Note that the words $w_j$ are well defined since, by the choice of $n$, we have $\sigma(j)2^z\le (k{-}1)2^z<2^{pn}$ for every $j\in[\![0,k{-}1]\!]$.

\begin{proposition}
\label{prop:jj'}
Suppose that $k>1$ and let $j,j'\in[\![0,k{-}1]\!]$. Then the word $w_j$ is accepted from $j'$ in the automaton $\Pi(\A_{m,2^p})$ if and only if $j=j'$. 
\end{proposition}

\begin{proof}
Let $n=|\rep_{2^p}\big((k{-}1)2^z\big)|$. Then $|w_j|=n$ for all $j\in[\![0,k{-}1]\!]$ and from Lemma~\ref{lem:k>1}, we know that $pn\ge z$. The result follows from the following computations:
\begin{align*}
j' 2^{p|w_j|} +\val_{2^p}(w_j)\equiv 0\pmod m
		&\iff	j' 2^{pn} + \sigma(j)2^z\equiv 0\pmod{k2^z}\\
		&\iff 	j' 2^{pn-z} + \sigma(j)\equiv 0\pmod k \\
		&\iff 	j' 2^{pn-z} -j2^{pn-z} \equiv 0\pmod k\\
		&\iff	j\equiv j' \pmod k \\
		&\iff	j=j'.
\end{align*}
\end{proof}

\begin{proposition}
\label{prop:jj'-bis}
Suppose that $k>1$ and let $j,j'\in[\![0,k{-}1]\!]$. Then the word $w_j\rep_{2^p}(m)$ is accepted from $j'$ in the automaton $\Pi(\A_{m,2^p})$ if and only if $j=j'$. 
\end{proposition}

\begin{proof}
Let $n=|\rep_{2^p}\big((k{-}1)2^z\big)|$, let $L=|\rep_{2^p}(m)|$ and, for each $j\in[\![0,k{-}1]\!]$, let $x_j=w_j\rep_{2^p}(m)$. From Lemma~\ref{lem:k>1}, we know that $pn\ge z$. Therefore, we have
\begin{align*}
j' 2^{p|x_j|} +\val_{2^p}(x_j)\equiv 0\pmod m
		&\iff 	j' 2^{p(n+L)} +\val_{2^p}(w_j)2^{pL}\equiv 0\pmod m \\
		&\iff	j' 2^{p(n+L)} + \sigma(j)2^{z+pL}\equiv 0\pmod{k2^z}\\
		&\iff 	j' 2^{p(n+L)-z} + \sigma(j)2^{pL}\equiv 0\pmod k \\
		&\iff 	j' 2^{pn-z} -j2^{pn-z} \equiv 0\pmod k\\
		&\iff	j\equiv j' \pmod k\\
		&\iff	j=j'
\end{align*}
and the result follows.
\end{proof}

\section{The product automaton $\A_{m,2^p}\times \A_{\T,2^p}$}

In this section, we study the product automaton $\A_{m,2^p} \times \A_{\T,2^p}$. Since the states of $\A_{m,2^p}$ are numbered from $0$ to $m{-}1$ and those of $\A_{\T,2^p}$ are $T$ and $B$, we denote the states of the product automaton by
\[
	(0,T),\ldots,(m{-}1,T) \andrm (0,B) , \ldots , (m{-}1, B).
\]
The transitions of $\A_{m,2^p} \times \A_{\T,2^p}$ are defined as follows. For $i,j\in[\![0,m{-}1]\!]$, $X,Y \in \{T,B\}$ and $d,e \in A_{2^p}$, there is a transition labeled by $(d,e)$ from the state $(i,X)$ to the state $(j,Y)$ if and only if
\[
	2^p i + e = md + j \quad \andrm \quad 
	Y = 	X_d.
\]
We denote by $\delta_\times$ the (partial) transition function of this product automaton. The state $(0,T)$ is both initial and final, and there is no other final state.


\begin{lemma}
\label{lem:transitionsProd}
For all $i,j\in[\![0,m{-}1]\!]$, $X ,Y \in \{T,B\}$ and $(u,v)\in (A_{2^p}\times A_{2^p})^*$, we have
$\delta_\times((i,X),(u,v))=(j,Y)$ if and only if
\[
	2^{p\,|(u,v)|}\, i + \val_{2^p}(v) = m\, \val_{2^p}(u) + j 
	\quad \andrm \quad
	Y= X_{\val_{2^p}(u)}.
\]
\end{lemma}

\begin{proof}
It suffices to combine Lemmas~\ref{lem:transitionsTM} and~\ref{lem:transitionsAmb}.
\end{proof}

In Figure~\ref{fig:product}, we have depicted the automaton $\A_{6,4} \times \A_{\T,4}$, as well as the automata $\A_{6,4}$ and $\A_{\T,4}$, which we have placed in such a way that the labels of the product automata can be easily deduced. For clarity, states are named $iX$ instead of $(i,X)$. We have drawn a full cycle in purple. It is of course not the only such cycle. This shows that the automaton $\A_{6,4} \times \A_{\T,4}$ is accessible and coaccessible. It will be true in general for the product automaton $\A_{m,2^p} \times \A_{\T,2^p}$. We give a proof of this fact below.

\begin{sidewaysfigure}
\begin{center}
\begin{minipage}{\linewidth} 
\begin{tikzpicture}[scale=0.8]
\tikzstyle{every node}=[shape=circle, fill=none, draw=black,
minimum size=30pt, inner sep=2pt]
\node(0) at (0,6.5) {$0$};
\node(1) at (3.5,6.5) {$1$};
\node(2) at (7,6.5) {$2$};
\node(3) at (10.5,6.5) {$3$};
\node(4) at (14,6.5) {$4$};
\node(5) at (17.5,6.5) {$5$};
\node(T) at (-5.5,0) {$T$};
\node(B) at (-5.5,-4.5) {$B$};
\node(0T) at (0,0) {$0T$};
\node(1T) at (3.5,0) {$1T$};
\node(2T) at (7,0) {$2T$};
\node(3T) at (10.5,0) {$3T$};
\node(4T) at (14,0) {$4T$};
\node(5T) at (17.5,0) {$5T$};
\node(0B) at (0,-4.5) {$0B$};
\node(1B) at (3.5,-4.5) {$1B$};
\node(2B) at (7,-4.5) {$2B$};
\node(3B) at (10.5,-4.5) {$3B$};
\node(4B) at (14,-4.5) {$4B$};
\node(5B) at (17.5,-4.5) {$5B$};
\tikzstyle{every node}=[shape=circle, fill=none, draw=black,
minimum size=25pt, inner sep=2pt]
\node at (0,0) {};
\node at (0,6.5) {};
\node at (-5.5,0) {};

\tikzstyle{etiquettedebut}=[very near start,rectangle,fill=black!20]
\tikzstyle{etiquettemilieu}=[midway,rectangle,fill=black!20]
\tikzstyle{every path}=[color=black, line width=0.5 pt]
\tikzstyle{every node}=[shape=circle, minimum size=5pt, inner sep=2pt]

\draw [->] (-1.5,6.5) to node {} (0); 
\draw [->] (0) to [loop above] node [left,rectangle,fill=black!20] {$(0,0)$} (0);
\draw [->] (1) to [loop below] node [left,rectangle,fill=black!20] {$(1,3)$} (1);
\draw [->] (2) to [loop left] node [left,rectangle,fill=black!20] {$(1,0)$} (2);
\draw [->] (3) to [loop below] node [left,rectangle,fill=black!20] {$(2,3)$} (3);
\draw [->] (4) to [loop above] node [left,rectangle,fill=black!20] {$(2,0)$} (4);
\draw [->] (5) to [loop above] node [left,rectangle,fill=black!20] {$(3,3)$} (5);
\draw [->] (0) to [bend left=15] node [sloped,etiquettemilieu] {$(0,1)$} (1);
\draw [->] (0) to [bend left=30] node [sloped,etiquettemilieu] {$(0,2)$} (2);
\draw [->] (0) to [bend left=45] node [sloped,etiquettedebut] {$(0,3)$} (3);
\draw [->] (1) to [bend left=30] node [sloped,pos=0.37,rectangle,fill=black!20] {$(0,0)$} (4);
\draw [->] (1) to [bend left=45] node [sloped,etiquettedebut] {$(0,1)$} (5);
\draw [->] (1) to [bend left=15] node [sloped,etiquettemilieu] {$(1,2)$} (0);
\draw [->] (2) to [bend left=15] node [sloped,etiquettemilieu] {$(1,1)$} (3);
\draw [->] (2) to [bend left=30] node [sloped,etiquettemilieu] {$(1,2)$} (4);
\draw [->] (2) to [bend left=45] node [sloped,etiquettemilieu] {$(1,3)$} (5);
\draw [->] (3) to [bend left=45] node [sloped,etiquettemilieu] {$(2,0)$} (0);
\draw [->] (3) to [bend left=40] node [sloped,pos=0.45,rectangle,fill=black!20] {$(2,1)$} (1);
\draw [->] (3) to [bend left=15] node [sloped,etiquettemilieu] {$(2,2)$} (2);
\draw [->] (4) to [bend left=15] node [sloped,etiquettemilieu] {$(2,1)$} (5);
\draw [->] (4) to [bend left=45] node [sloped,etiquettemilieu] {$(3,2)$} (0);
\draw [->] (4) to [bend left=40] node [sloped,pos=0.5,rectangle,fill=black!20] {$(3,3)$} (1);
\draw [->] (5) to [bend left=40] node [sloped,pos=0.45,rectangle,fill=black!20] {$(3,0)$} (2);
\draw [->] (5) to [bend left=35] node [sloped,etiquettemilieu] {$(3,1)$} (3);
\draw [->] (5) to [bend left=15] node [sloped,etiquettemilieu] {$(3,2)$} (4);

\draw [->] (-7,0) to node {} (T); 
\draw [->] (T) to [loop above] node [etiquettemilieu] {\begin{tabular}{c}
$(0,0),(0,1),(0,2),(0,3)$ \\
$(3,0),(3,1),(3,2),(3,3)$ \\
\end{tabular}} (T);
\draw [->] (B) to [loop below] node [etiquettemilieu] {\begin{tabular}{c}
$(0,0),(0,1),(0,2),(0,3)$ \\
$(3,0),(3,1),(3,2),(3,3)$ \\
\end{tabular}} (B);
\draw [->] (T) to [bend right=30] node [left=-0.3cm,etiquettemilieu] {\begin{tabular}{c}
$(1,0),(1,1),$\\$(1,2),(1,3)$ \\
$(2,0),(2,1),$\\$(2,2),(2,3)$ \\
\end{tabular}} (B);
\draw [->] (B) to [bend right=30] node [right=-0.3cm,etiquettemilieu] {\begin{tabular}{c}
$(1,0),(1,1),$\\$(1,2),(1,3)$ \\
$(2,0),(2,1),$\\$(2,2),(2,3)$ \\
\end{tabular}} (T);

\draw [->] (-1.5,0) to node {} (0T); 
\draw [->] (0T) to [loop above] node [] {} (0T);
\draw [->] (5T) to [loop above] node [] {} (5T);
\draw [->] (0B) to [loop below] node [] {} (0B);
\draw [->] (5B) to [loop above] node [] {} (5B);
\draw [->] (0T) to [] node [] {} (1T);
\draw [purple,->] (0T) to [bend left=20] node [] {} (2T);
\draw [->] (0T) to [bend left=25] node [] {} (3T);
\draw [->] (1T) to [bend left=20] node [] {} (4T);
\draw [->] (1T) to [bend left=30] node [] {} (5T);
\draw [purple,->] (1T) to [] node [] {} (0B);
\draw [->] (1T) to [bend left=15] node [] {} (1B);
\draw [->] (2T) to [bend left=15] node [] {} (2B);
\draw [->] (2T) to [bend left=5] node [] {} (3B);
\draw [purple,->] (2T) to [] node [] {} (4B);
\draw [->] (2T) to [] node [] {} (5B);
\draw [->] (3T) to [] node [] {} (0B);
\draw [->] (3T) to [] node [] {} (1B);
\draw [->] (3T) to [bend left=5] node [] {} (2B);
\draw [purple,->] (3T) to [bend left=15] node [] {} (3B);
\draw [->] (4T) to [bend left=15] node [] {} (4B);
\draw [->] (4T) to [] node [] {} (5B);
\draw [->] (4T) to [bend right=35] node [] {} (0T);
\draw [purple,->] (4T) to [bend right=25] node [] {} (1T);
\draw [->] (5T) to [bend right=25] node [] {} (2T);
\draw [->] (5T) to [bend right=20] node [] {} (3T);
\draw [purple,->] (5T) to [] node [] {} (4T);

\draw [purple,->] (0B) to [] node [] {} (1B);
\draw [->] (0B) to [bend right=20] node [] {} (2B);
\draw [->] (0B) to [bend right=25] node [] {} (3B);
\draw [->] (1B) to [bend right=20] node [] {} (4B);
\draw [purple,->] (1B) to [bend right=30] node [] {} (5B);
\draw [->] (1B) to [] node [] {} (0T);
\draw [->] (1B) to [bend left=15] node [] {} (1T);
\draw [->] (2B) to [bend left=15] node [] {} (2T);
\draw [purple,->] (2B) to [bend left=5] node [] {} (3T);
\draw [->] (2B) to [] node [] {} (4T);
\draw [->] (2B) to [] node [] {} (5T);
\draw [purple,->] (3B) to [] node [] {} (0T);
\draw [->] (3B) to [] node [] {} (1T);
\draw [->] (3B) to [bend left=5] node [] {} (2T);
\draw [->] (3B) to [bend left=15] node [] {} (3T);
\draw [->] (4B) to [bend left=15] node [] {} (4T);
\draw [purple,->] (4B) to [] node [] {} (5T);
\draw [->] (4B) to [bend left=35] node [] {} (0B);
\draw [->] (4B) to [bend left=25] node [] {} (1B);
\draw [purple,->] (5B) to [bend left=25] node [] {} (2B);
\draw [->] (5B) to [bend left=20] node [] {} (3B);
\draw [->] (5B) to [] node [] {} (4B);
\end{tikzpicture}
\caption{The product automaton $\A_{6,4} \times \A_{\T,4}$}
\label{fig:product}
\end{minipage}
\end{center}
\end{sidewaysfigure}
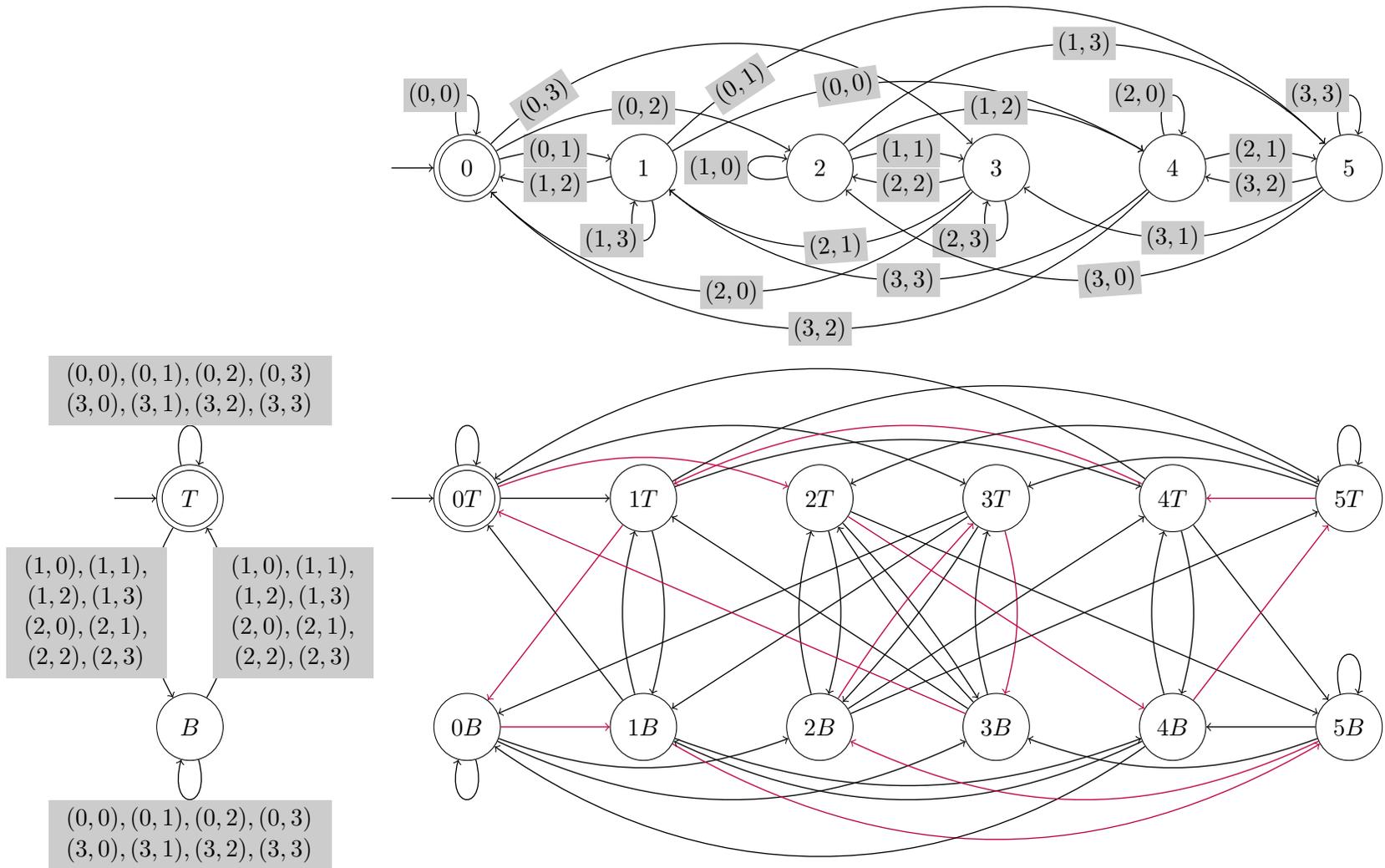

\begin{corollary}
\label{cor:0Baccessible}
The word $\rep_{2^p}(1,m)$ is accepted from the state $(0,B)$ in $\A_{m,2^p} \times \A_{\T,2^p}$. In particular, the state $(0,B)$ is coaccessible in $\A_{m,2^p} \times \A_{\T,2^p}$.
\end{corollary}

\begin{proof}
This follows from Lemma~\ref{lem:transitionsProd}.
\end{proof}

\begin{lemma}
\label{lem:proddisj1}
For each $i \in \{0 , \ldots , m{-}1 \}$, the states $(i,T)$ et $(i,B)$ of the automaton $\A_{m,2^p} \times \A_{\T,2^p}$ are disjoint.
\end{lemma}

\begin{proof}
This comes from the fact that $\A_{\T,2^p}$ has disjoint states.
\end{proof}

\begin{lemma}
\label{lem:proddisj2}
For distinct $i,j \in [\![0,m{-}1]\!]$ and for $X,Y \in\{T,B\}$, the states $(i,X)$ et $(j,Y)$ are disjoint in $\A_{m,2^p} \times \A_{\T,2^p}$.
\end{lemma}

\begin{proof}
Let $i,j \in [\![0,m{-}1]\!]$ and $X,Y \in\{T,B\}$. Suppose that there exists a word $(u,v)\in(A_{2^p}\times A_{2^p})^*$ which is accepted from both $(i,X)$ and $(j,Y)$ in $\A_{m,2^p} \times \A_{\T,2^p}$. Then $(u,v)$ is accepted from both $i$ and $j$ in $\A_{m,2^p}$. Since the automaton $\A_{m,2^p}$ has disjoint states by Proposition~\ref{prop:Ambproperties}, this implies that $i=j$.
\end{proof}

We are now ready to establish the main properties of the product automaton $\A_{m,2^p} \times \A_{\T,2^p}$.

\begin{proposition}
\label{prop:prodaccessible}
The automaton $\A_{m,2^p} \times \A_{\T,2^p}$ is complete, accessible, coaccessible and has disjoint states. In particular, it is the minimal automaton of $\val_{2^p}^{-1}(\{(t,mt)\colon t\in\T\})$.
\end{proposition}

\begin{proof}
By construction of the product automaton and since
\[
	\{(n,mn)\colon n\in \N\}
	\cap \big(\T\times \N\big)
	=\{(t,mt)\colon t\in\T\},
\]
we get that the product automaton $\A_{m,2^p} \times \A_{\T,2^p}$ accepts the language
\[
	\val_{2^p}^{-1}(\{(t,mt)\colon t\in\T\}).
\]
By Lemma~\ref{lem:transitionsProd}, we can check that for every $i\in[\![0,m{-}1]\!]$, the states $(i,T)$ and $(i,B)$ are  accessible thanks to the word $\rep_{2^p}(0,i)$ and $\rep_{2^p}(1,m+i)$ respectively. Hence, $\A_{m,2^p} \times \A_{\T,2^p}$ is accessible. To show the coaccessibility, we now fix some $i\in[\![0,m{-}1]\!]$ and $X \in \{T,B \}$. By Proposition~\ref{prop:Ambproperties}, we already know that the automaton $\A_{m,2^p}$ is coaccessible. Therefore, we can find $(u,v)\in(A_{2^p}\times A_{2^p})^*$ such that there is a path labeled by $(u,v)$ from $i$ to $0$ in $\A_{m,2^p}$. Thus, by reading $(u,v)$ from the state $(i,X)$ in $\A_{m,2^p} \times \A_{\T,2^p}$, we reach either the state $(0,T)$ or the state $(0,B)$. If we reach $(0,T)$, then the state $(i,X)$ is coaccessible. If we reach $(0,B)$ instead, then we may apply Corollary~\ref{cor:0Baccessible} in order to obtain that $(i,X)$ is coaccessible as well. Finally, in order to see that $\A_{m,2^p} \times \A_{\T,2^p}$ has disjoint states, it suffices to combine Lemmas~\ref{lem:proddisj1} and~\ref{lem:proddisj2}. 
\end{proof}

\section{The projection $\Pi \left( \A_{m,2^p} \times \A_{\T,2^p} \right)$ of the product automaton}

The aim of this section is to provide a DFA accepting the language $\val_{2^p}^{-1}(m\T)$. This automaton is denoted by $\Pi \left( \A_{m,2^p} \times \A_{\T,2^p} \right)$ and is defined from the automaton $\A_{m,2^p} \times \A_{\T,2^p}$ by only considering the second component of each label. Formally, the states of $\Pi \left( \A_{m,2^p} \times \A_{\T,2^p} \right)$ are 
\[
	(0,T), \ldots , (m{-}1, T) \andrm (0,B) , \ldots , (m{-}1, B),
\]
the state $(0,T)$ is both initial and final and no other state is final, and the transitions are defined as follows. For $i,j \in [\![0,m{-}1]\!]$, $X,Y \in \{T,B\}$ and $e\in A_{2^p}$, there is a transition labeled by $e$ from the state $(i,X)$ to the state $(j,Y)$ if and only if there exists $d \in A_{2^p}$ such that
\[
	2^p i + e = md + j \quad \andrm \quad 
	Y = 	X_d.
\]

\begin{example}
The automata $\A_{6,4} \times \A_{\T,4}$ and $\Pi \left( \A_{6,4} \times \A_{\T,4} \right)$ are depicted in Figures~\ref{fig:product} and~\ref{fig:projected-aut} respectively. In Figure~\ref{fig:projected-aut}, all edges labeled by $0$ ($1,2$ and $3$ respectively) are represented in black (blue, red and green respectively).
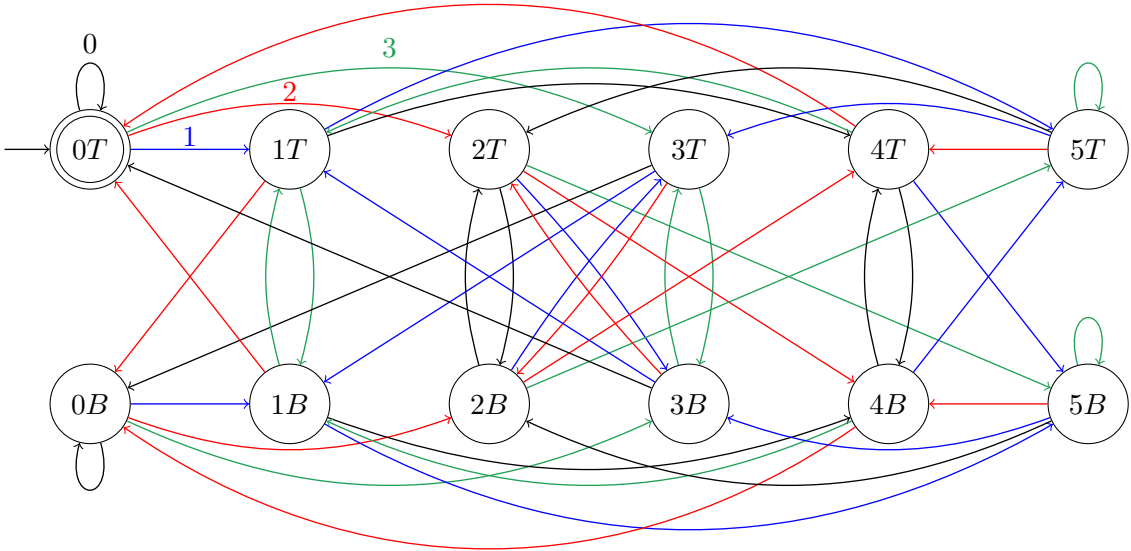
\begin{figure}[htb]
\centering
\begin{tikzpicture}[scale=0.75]
\tikzstyle{every node}=[shape=circle, fill=none, draw=black,
minimum size=30pt, inner sep=2pt]
\node(0T) at (0,0) {$0T$};
\node(1T) at (3.5,0) {$1T$};
\node(2T) at (7,0) {$2T$};
\node(3T) at (10.5,0) {$3T$};
\node(4T) at (14,0) {$4T$};
\node(5T) at (17.5,0) {$5T$};
\node(0B) at (0,-4.5) {$0B$};
\node(1B) at (3.5,-4.5) {$1B$};
\node(2B) at (7,-4.5) {$2B$};
\node(3B) at (10.5,-4.5) {$3B$};
\node(4B) at (14,-4.5) {$4B$};
\node(5B) at (17.5,-4.5) {$5B$};
\tikzstyle{every node}=[shape=circle, fill=none, draw=black,
minimum size=25pt, inner sep=2pt]
\node at (0,0) {};

\tikzstyle{etiquettedebut}=[very near start,rectangle,fill=black!20]
\tikzstyle{etiquettemilieu}=[midway,rectangle,fill=black!20]
\tikzstyle{every path}=[color=black, line width=0.5 pt]
\tikzstyle{every node}=[shape=circle, minimum size=5pt, inner sep=2pt]
\draw [->] (-1.5,0) to node {} (0T); 
\draw [->] (0T) to [loop above] node [] {$0$} (0T);
\draw [green,->] (5T) to [loop above] node [] {} (5T);
\draw [->] (0B) to [loop below] node [] {} (0B);
\draw [green,->] (5B) to [loop above] node [] {} (5B);
\draw [blue,->] (0T) to [] node [above=-0.1] {$1$} (1T);
\draw [red,->] (0T) to [bend left=20] node [above=-0.1] {$2$} (2T);
\draw [green,->] (0T) to [bend left=25] node [above] {$3$} (3T);
\draw [->] (1T) to [bend left=20] node [] {} (4T);
\draw [blue,->] (1T) to [bend left=30] node [] {} (5T);
\draw [red,->] (1T) to [] node [] {} (0B);
\draw [green,->] (1T) to [bend left=15] node [] {} (1B);
\draw [->] (2T) to [bend left=15] node [] {} (2B);
\draw [blue,->] (2T) to [bend left=5] node [] {} (3B);
\draw [red,->] (2T) to [] node [] {} (4B);
\draw [green,->] (2T) to [] node [] {} (5B);
\draw [->] (3T) to [] node [] {} (0B);
\draw [blue,->] (3T) to [] node [] {} (1B);
\draw [red,->] (3T) to [bend left=5] node [] {} (2B);
\draw [green,->] (3T) to [bend left=15] node [] {} (3B);
\draw [->] (4T) to [bend left=15] node [] {} (4B);
\draw [blue,->] (4T) to [] node [] {} (5B);
\draw [red,->] (4T) to [bend right=35] node [] {} (0T);
\draw [green,->] (4T) to [bend right=25] node [] {} (1T);
\draw [->] (5T) to [bend right=25] node [] {} (2T);
\draw [blue,->] (5T) to [bend right=20] node [] {} (3T);
\draw [red,->] (5T) to [] node [] {} (4T);
\draw [blue,->] (0B) to [] node [] {} (1B);
\draw [red,->] (0B) to [bend right=20] node [] {} (2B);
\draw [green,->] (0B) to [bend right=25] node [] {} (3B);
\draw [->] (1B) to [bend right=20] node [] {} (4B);
\draw [blue,->] (1B) to [bend right=30] node [] {} (5B);
\draw [red,->] (1B) to [] node [] {} (0T);
\draw [green,->] (1B) to [bend left=15] node [] {} (1T);
\draw [->] (2B) to [bend left=15] node [] {} (2T);
\draw [blue,->] (2B) to [bend left=5] node [] {} (3T);
\draw [red,->] (2B) to [] node [] {} (4T);
\draw [green,->] (2B) to [] node [] {} (5T);
\draw [->] (3B) to [] node [] {} (0T);
\draw [blue,->] (3B) to [] node [] {} (1T);
\draw [red,->] (3B) to [bend left=5] node [] {} (2T);
\draw [green,->] (3B) to [bend left=15] node [] {} (3T);
\draw [->] (4B) to [bend left=15] node [] {} (4T);
\draw [blue,->] (4B) to [] node [] {} (5T);
\draw [red,->] (4B) to [bend left=35] node [] {} (0B);
\draw [green,->] (4B) to [bend left=25] node [] {} (1B);
\draw [->] (5B) to [bend left=25] node [] {} (2B);
\draw [blue,->] (5B) to [bend left=20] node [] {} (3B);
\draw [red,->] (5B) to [] node [] {} (4B);
\end{tikzpicture}
\caption{The projected automaton $\Pi \left( \A_{6,4} \times \A_{\T,4} \right)$.}
\label{fig:projected-aut}
\end{figure}
\end{example}

\begin{lemma}
\label{lem:projdisj}
For every $i \in [\![0,m{-}1]\!]$, the states $(i,T)$ and $(i,B)$ are disjoint in the projected automaton $\Pi \left( \A_{m,2^p} \times \A_{\T,2^p} \right)$.
\end{lemma}

\begin{proof}
Let $i \in [\![0,m{-}1]\!]$. It follows from Remark~\ref{rem:unicitelettrepremcomp} and the definitions of the transition functions of  $\A_{m,2^p} \times \A_{\T,2^p}$ and $\Pi \left( \A_{m,2^p} \times \A_{\T,2^p} \right)$ that if a word $v$ over $A_{2^p}$ is accepted from both $(i,T)$ and $(i,B)$ in $\Pi \left( \A_{m,2^p} \times \A_{\T,2^p} \right)$,  then there exists a unique word $u$ over $A_{2^p}$ of length $|v|$ such that the word $(u,v)$ is accepted from both $(i,T)$ and $(i,B)$ in $\A_{m,2^p} \times \A_{\T,2^p}$. The conclusion then follows from Lemma~\ref{lem:proddisj1}.
\end{proof}

\begin{proposition}
\label{prop:m-odd}
The automaton $\Pi \left(\A_{m,2^p}\times \A_{\T,2^p} \right)$
\begin{itemize}
\item accepts $\val_{2^p}^{-1}(m\T)$
\item is deterministic
\item is complete
\item is accessible
\item is coaccessible
\item has disjoint states if $m$ is odd 
\item is minimal if $m$ is odd. 
\end{itemize}
\end{proposition}

\begin{proof}
By construction, $\Pi \left( \A_{m,2^p} \times \A_{\T,2^p} \right)$ accepts $\val_{2^p}^{-1}(m\T)$; see Section~\ref{sec:objectifmethode}. The fact that this automaton is deterministic and complete follows from Remark~\ref{rem:unicitelettrepremcomp}. It is accessible and coaccessible because so is $\A_{m,2^p} \times \A_{\T,2^p}$.  Now we turn to the last two items. If a word $v$ over $A_{2^p}$ is accepted from some state $(i,X)$ in $\Pi \left( \A_{m,2^p} \times \A_{\T,2^p} \right)$,  then there exists a word $u$ over $A_{2^p}$ of length $|v|$ such that the word $(u,v)$ is accepted from $(i,X)$ in $\A_{m,2^p} \times \A_{\T,2^p}$. We deduce that $(u,v)$ is accepted from the state $i$ in $\A_{m,2^p}$ and in turn, that $v$ is accepted from the state $i$ in $\Pi \left( \A_{m,2^p} \right)$. Therefore, and by combining Proposition~\ref{prop:projAmbdisj} and Lemma~\ref{lem:projdisj}, we obtain  that if $m$ is odd then the automaton $\Pi \left( \A_{m,2^p} \times \A_{\T,2^p} \right)$ has disjoint states. It directly follows that $\Pi \left( \A_{m,2^p} \times \A_{\T,2^p} \right)$ is minimal if $m$ is odd.
\end{proof}

\begin{corollary}
\label{cor:impair}
If $m$ is odd, then the state complexity of $m \T $ with respect to the base $2^p$ is $2m$.
\end{corollary}

Note that Corollary~\ref{cor:impair} and Theorem~\ref{thm:main} are consistent in the case where $m$ is odd, i.e.\  where $z=0$. However, we will see in the next section that the DFA $\Pi \left( \A_{m,2^p} \times \A_{\T,2^p} \right)$ is never minimal for even $m$ because it contains several states accepting the same language.

\section{Minimization of $\Pi \left( \A_{m,2^p} \times \A_{\T,2^p} \right)$}

We start by defining some classes of states of $\Pi \left( \A_{m,2^p} \times \A_{\T,2^p} \right)$. Our aim is twofold. First, we will prove that those subsets consist in {\em indistinguishable} states, i.e.\ accepting the same language. Second, we will show that states belonging to different such subsets are {\em distinguishable}, i.e.\ accepts different languages. Otherwise stated, these classes correspond to the left quotients $w^{{-}1}L$ where $w$ is any word over the alphabet $A_{2^p}$ and $L=\val^{-1}_{2^p}(m\T)$.

\begin{definition}
For $(j,X) \in \big([\![1,k{-}1]\!]\times \{T,B\}\big)\cup \{(0,B)\}$, we define
\[
	[(j,X)] = \{ ( j+k\ell, X_\ell) \colon 0 \le \ell \le 2^z {-}1\}\\
\]
and $[(0,T)] = \{ (0,T)\}$. We say that $[(j,X)]$ is the {\em class} of the state $(j,X)$. 
\end{definition}

\begin{remark}
Note that the classes $[(j,X)]$ are pairwise disjoint: $[(j,X)]\cap[(j',X')]=\emptyset$ if $(j,X)\ne (j',X')$. If $m$ is odd, i.e.\ if $z =0$, then all these classes are reduced to a single state. If $m$ is a power of $2$, i.e.\  if $k=1$, then there is no class of the form $[(j,X)]$ with $j\ge 1$. 
\end{remark}

\begin{definition}
For $\alpha \in [\![0,z{-}1]\!]$, we define a {\em pre-class} $C_\alpha$ of size $2^\alpha$:
\[
	C_\alpha = [(k2^{z-\alpha-1},B)] 
	= \{(k2^{z-\alpha-1} + k2^{z-\alpha} \ell,B_\ell) \colon \ell\in [\![0,2^\alpha{-}1]\!]\}.
\]
Then, for $\beta\in [\![0,\lceil\frac zp\rceil{-}2]\!]$, we define a {\em class} $\Gamma_\beta$ as follows:
\[
	\Gamma_\beta = \bigcup_{\alpha =\beta p}^{\beta p+p-1} C_\alpha.
\]
In addition, we set
\[
	\Gamma_{\lceil \frac zp \rceil - 1} 
	= \bigcup_{\alpha = \left(\left\lceil\frac zp \right\rceil-1\right) p}^{z-1} C_{\alpha}.
\]
\end{definition}

\begin{remark}
Note that the classes $\Gamma_\beta$ are pairwise disjoint. If $m$ is odd, i.e.\ if $z =0$, then there is no such class $\Gamma_\beta$. 
\end{remark}

\begin{remark}
If a class $[(j,X)]$ or $\Gamma_\beta$ exists, then it is nonempty. Moreover, the classes $\Gamma_\beta$ together with the class $[(0,T)]$ form a partition of $ \{(k\ell,T_\ell) \colon \ell\in[\![0,2^z{-}1]\!]\}$. Therefore, the classes $[(j,X)]$ and $\Gamma_\beta$ form a partition of the set of states of $\Pi \left( \A_{m,2^p} \times \A_{\T,2^p} \right)$.
\end{remark}

\begin{example}
For $m=24$ and $p=2$, the classes defined above are
\begin{align*}
	[(0,T)] & =\{(0,T)\} \\
	[(1,T)] & =\{(1,T),(4,B),(7,B),(10,T),(13,B),(16,T),(19,T),(22,B)\} \\
	[(2,T)] & =\{(2,T),(5,B),(8,B),(11,T),(14,B),(17,T),(20,T),(23,B)\} \\
	[(0,B)] & =\{(0,B),(3,T),(6,T),(9,B),(12,T),(15,B),(18,B),(21,T)\} \\
	[(1,B)] & =\{(1,B),(4,T),(7,T),(10,B),(13,T),(16,B),(19,B),(22,T)\} \\
	[(2,B)] & =\{(2,B),(5,T),(8,T),(11,B),(14,T),(17,B),(20,B),(23,T)\} \\
	\Gamma_0 & =C_0\cup C_1= \{(12,B)\}\cup\{(6,B),(18,T)\}=\{(6,B),(12,B),(18,T)\} \\
	\Gamma_1 & =C_2= \{(3,B),(9,T),(15,T),(21,B)\}.
\end{align*}
In Figure~\ref{fig:projected-aut24}, the states of the automaton $\Pi \left( \A_{24,4} \times \A_{\T,4} \right)$ are colored with respect to these classes.
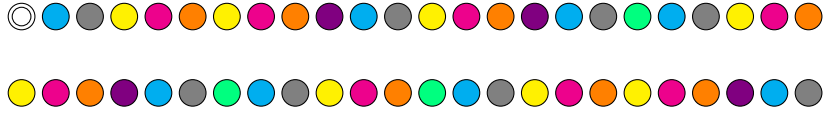
\begin{figure}
\centering
\begin{tikzpicture}
[scale=0.225]
\tikzstyle{every node}=[shape=circle, fill=none, draw=black,
minimum size=10pt, inner sep=2pt]
\node(0T) at (0,0) {};
\node[fill=cyan](1T) at (2,0) {};
\node[fill=gray](2T) at (4,0) {};
\node[fill=yellow](3T) at (6,0) {};
\node[fill=magenta](4T) at (8,0) {};
\node[fill=orange](5T) at (10,0) {};
\node[fill=yellow](6T) at (12,0) {};
\node[fill=magenta](7T) at (14,0) {};
\node[fill=orange](8T) at (16,0) {};
\node[fill=violet](9T) at (18,0) {};
\node[fill=cyan](10T) at (20,0) {};
\node[fill=gray](11T) at (22,0) {};
\node[fill=yellow](12T) at (24,0) {};
\node[fill=magenta](13T) at (26,0) {};
\node[fill=orange](14T) at (28,0) {};
\node[fill=violet](15T) at (30,0) {};
\node[fill=cyan](16T) at (32,0) {};
\node[fill=gray](17T) at (34,0) {};
\node[fill=vert](18T) at (36,0) {};
\node[fill=cyan](19T) at (38,0) {};
\node[fill=gray](20T) at (40,0) {};
\node[fill=yellow](21T) at (42,0) {};
\node[fill=magenta](22T) at (44,0) {};
\node[fill=orange](23T) at (46,0) {};

\node[fill=yellow](0B) at (0,-4.5) {};
\node[fill=magenta](1B) at (2,-4.5) {};
\node[fill=orange](2B) at (4,-4.5) {};
\node[fill=violet](3B) at (6,-4.5) {};
\node[fill=cyan](4B) at (8,-4.5) {};
\node[fill=gray](5B) at (10,-4.5) {};
\node[fill=vert](6B) at (12,-4.5) {};
\node[fill=cyan](7B) at (14,-4.5) {};
\node[fill=gray](8B) at (16,-4.5) {};
\node[fill=yellow](9B) at (18,-4.5) {};
\node[fill=magenta](10B) at (20,-4.5) {};
\node[fill=orange](11B) at (22,-4.5) {};
\node[fill=vert](12B) at (24,-4.5) {};
\node[fill=cyan](13B) at (26,-4.5) {};
\node[fill=gray](14B) at (28,-4.5) {};
\node[fill=yellow](15B) at (30,-4.5) {};
\node[fill=magenta](16B) at (32,-4.5) {};
\node[fill=orange](17B) at (34,-4.5) {};
\node[fill=yellow](18B) at (36,-4.5) {};
\node[fill=magenta](19B) at (38,-4.5) {};
\node[fill=orange](20B) at (40,-4.5) {};
\node[fill=violet](21B) at (42,-4.5) {};
\node[fill=cyan](22B) at (44,-4.5) {};
\node[fill=gray](23B) at (46,-4.5) {};

\tikzstyle{every node}=[shape=circle, fill=none, draw=black,
minimum size=7pt, inner sep=2pt]
\node at (0,0) {};
\end{tikzpicture}
\caption{The classes of the projected automaton $\Pi \left( \A_{24,4} \times \A_{\T,4} \right)$.}
\label{fig:projected-aut24}
\end{figure}
\end{example}

\subsection{States of the same class are indistinguishable}
\label{sec:reduction1}

In order to prove that two states $(j,X)$ and $(j',X')$ of the automaton $\Pi \left( \A_{m,2^p} \times \A_{\T,2^p} \right)$ are indistinguishable, we have to prove that $L_{(j,X)}=L_{(j',X')}$. The general procedure that we use for proving that $L_{(j,X)}\subseteq L_{(j',X')}$ goes as follows. Let $v\in L_{(j,X)}$ and let $n=|v|$. Then we know that there exists a word $u$ over $A_{2^p}$ of length $|v|$ such that $(u,v)$ is accepted from the state $(j,X)$ in $\A_{m,2^p} \times \A_{\T,2^p}$ (before the projection). If $d=\val_{2^p}(u)$ and $e = \val_{2^p} (v)$, then, in view of Lemma~\ref{lem:transitionsProd}, we must have 
\[
	2^{pn}j+e = md \quad \andrm \quad X_d=T
\]
(the only final state of $\A_{m,2^p} \times \A_{\T,2^p}$ is $(0,T)$). Moreover, since $n=|v|$, we have $d,e\in[\![0,2^{pn}{-}1]\!]$. Now, in order to prove that $v\in L_{(j',X')}$, we have to find a word $u'$ over $A_{2^p}$ of length $n$ such that $(u',v)$ is accepted from $(j',X')$ in $\A_{m,2^p} \times \A_{\T,2^p}$. But then, we necessarily have that 
\[
	\val_{2^p}(u')=\frac{2^{pn}j'+e}{m}.
\] 
Let thus $d'=\frac{2^{pn}j+e}{m}$. We obtain that $v\in L_{(j',X')}$ if and only if $d'\in[\![0,2^{pn}{-}1]\!]$ and $X'_{d'}=T$. Indeed, in this case, $|\rep_{2^p}(d')|\le n$ and thus, we can take the word $u'=0^{n-|\rep_{2^p}(d')|}\rep_{2^p}(d')$.
\medskip

First, we show that two states of the same class of the form $[(j,X)]$ are indistinguishable.

\begin{proposition}
\label{prop:jX}
Let $j \in [\![1,k{-}1]\!]$, $X \in \{T,B\}$ and $\ell \in [\![0,2^z{-}1]\!]$. We have
\[
	L_{( j,X)} = L_{( j+k\ell, X_\ell)} 
\]
in $\Pi \left( \A_{m,2^p} \times \A_{\T,2^p} \right)$.
\end{proposition}

\begin{proof}
Let $v\in A_{2^p}^*$, $n=|v|$, $e=\val_{2^p}(v)$, $d=\frac{2^{pn}j+e}{m}$ and $d'=\frac{2^{pn}(j+k\ell)+e}{m}$.
We have to prove that $d\in[\![0,2^{pn}{-}1]\!]$ and $X_d=T$ if and only if $d'\in[\![0,2^{pn}{-}1]\!]$ and $(X_\ell)_{d'}=T$. 

Since $1\le j <k$ and $0\le e<2^{pn}$, we have
\begin{equation}
\label{eqn:jX}
	0< d=\frac{2^{pn}j+e}{m}<\frac{2^{pn}k}{m}=2^{pn-z}.
\end{equation}
Since $d'= d + \frac{2^{pn}k\ell}{m} =d+2^{pn-z}\ell$, it follows from~\eqref{eqn:jX} that if $d$ and $d'$ are both integers, then we must have
\[
	\rep_2(d')=\rep_2(\ell)0^{pn-z-|\rep_2(d)|}\rep_2(d).
\] 
Therefore, $d\in\T$ if and only if either $\ell\in\T\andrm d'\in\T$, or $\ell\notin\T\andrm d'\notin\T$, and hence  $X_d=(X_\ell)_{d'}$. 

Now, suppose that $d\in[\![0,2^{pn}{-}1]\!]$ and $X_d=T$. It follows from~\eqref{eqn:jX} that $pn> z$, for otherwise we would have $0<d<1$, which is not possible since $d$ is an integer. Therefore, we get that $d' = d +2^{pn-z}\ell$ is a positive integer. We also get from~\eqref{eqn:jX} that
\[
	d' = d+2^{pn-z}\ell<2^{pn-z} (\ell+1)\le  2^{pn}.
\]
Consequently, $d'\in[\![0,2^{pn}{-}1]\!]$ and $(X_\ell)_{d'}=X_d=T$. 

Conversely, suppose that $d'\in[\![0,2^{pn}{-}1]\!]$ and $(X_\ell)_{d'}=T$. In view of~\eqref{eqn:jX} and since $d= d'-2^{pn-z}\ell$, in order to obtain that $d\in[\![0,2^{pn}{-}1]\!]$, it is enough to show that $pn> z$. Proceed by contradiction and suppose that $pn \le z$. Let $q=\big\lfloor\frac{\ell}{2^{z-pn}}\big\rfloor$. On the one hand, since $j\ge 1$ and $e\ge 0$, we obtain
\[
	d'=\frac{2^{pn}(j+k\ell)+e}{m}>\frac{2^{pn}k\ell}{m}=\frac{\ell}{2^{z-pn}}\ge q.
\]
On the other hand, since $\ell\le (q+1)2^{z-pn}{-}1$, $e<2^{pn}$ and $j\le k{-}1$, we obtain
\[
	d' 	< \frac{2^{pn}(j+k(q+1)2^{z-pn}-k)+2^{pn}}{m} 
		= q+1+2^{pn}\frac{j-k+1}{m} 
		\le q+1.
\]
This is not possible since $d'$ is an integer, and hence $pn> z$. Consequently, $d\in[\![0,2^{pn}{-}1]\!]$ and $X_d=T$ as desired. 
\end{proof}

Since the class $[(0,T)]$ is a singleton, the only left case to consider is that of $[(0,B)]$.

\begin{proposition}
Let $\ell \in [\![1,2^z{-}1]\!]$. We have
\[
	L_{(0,B)} = L_{( k\ell, B_\ell)}
\]
in $\Pi \left( \A_{m,2^p} \times \A_{\T,2^p} \right)$.
\end{proposition}

\begin{proof}
Let $v\in A_{2^p}^*$, $n=|v|$, $e=\val_{2^p} (v)$, $d=\frac em$ and $d' = \frac{2^{pn}k\ell +e}{m}$. We have to prove that we have $d\in[\![0,2^{pn}{-}1]\!]$ and $B_d=T$ if and only if $d'\in[\![0,2^{pn}{-}1]\!]$ and $(B_\ell)_{d'}=T$. 

Since $0\le e<2^{pn}$, we have
\begin{equation}
\label{eqn:0B}
 	0\le d = \frac em < \frac{2^{pn}}{m} = \frac{2^{pn-z}}{k}.
\end{equation}
Since $k\ge 1$, it follows that $d< 2^{pn-z}$ and we get that $B_d=(B_\ell)_{d'}$ as in the proof of Proposition~\ref{prop:jX}, provided that both $d$ and $d'$ are integers. 

Now, suppose that $d\in[\![0,2^{pn}{-}1]\!]$ and $B_d=T$, that is, that $d$ is an integer and $d\notin\T$. If $pn \le z$ then we get from~\eqref{eqn:0B} that $0\le d<1$. But since $d$ is an integer, this implies that $d=0$, which is impossible because $d\notin\T$. Thus, $pn > z$ and $d'=d+\ell 2^{pn-z}$ is a nonnegative integer. Moreover, we have
\[
	d' = d+\ell 2^{pn-z} < \frac{2^{pn-z}}{k}+ (2^z{-}1) 2^{pn-z} = 2^{pn} + 2^{pn-z} \left(\frac{1}{k} - 1\right) \le 2^{pn}.
\]
Hence $d'\in[\![0,2^{pn}{-}1]\!]$ and $(B_\ell)_{d'}=B_d=T$.

Conversely, suppose that $d'\in[\![0,2^{pn}{-}1]\!]$ and $(B_\ell)_{d'}=T$. In particular, we have $d'\in\T\iff\ell\notin \T$. From~\eqref{eqn:0B}, we know that $0\le d< 2^{pn}$. We claim that $pn> z$. Proceed by contradiction and suppose that $pn \le z$. Let $q=\left\lfloor\frac{\ell}{2^{z-pn}}\right\rfloor$. Then, on the one hand, we have
\[
	d'=\frac{2^{pn}k\ell+e}{m}\ge \frac{\ell}{2^{z-pn}}\ge q.
\]
On the other hand, since $\ell\le (q+1)2^{z-pn}{-}1$ and $e<2^{pn}$, we obtain
\[
	d'=\frac{2^{pn}k\ell+e}{m}
		< \frac{2^{pn}k((q+1)2^{z-pn}{-}1)+2^{pn}}{m}
		= q+1-\frac{2^{pn}(k{-}1)}{m}
		\le (q+1),
\]
and hence $d'<q+1$. Since $d'$ is an integer, we get that $d'=q$, $e=0$ and $\ell=2^{z-pn}d'$. But then we would have
\[
	\rep_2(\ell)=\rep_2(d')0^{z-pn},
\]
contradicting that $d'\in\T\iff\ell\notin \T$. Thus $pn> z$ and  $d = d'-\ell 2^{pn-z}$ is an integer. Altogether, we get that $d\in[\![0,2^{pn}{-}1]\!]$ and $B_d=T$.
\end{proof}

\begin{corollary}\label{cor:0B}
For each $(j,X)\in[\![0,k{-}1]\!]\times\{T,B\}$, all states of the class $[(j,X)]$ are indistinguishable in $\Pi \left( \A_{m,2^p} \times \A_{\T,2^p} \right)$.
\end{corollary}

Now, we show that two states of the same class of the form $\Gamma_\beta$ are indistinguishable.

\begin{proposition}
Suppose that $z\ge 1$ and let $\alpha\in[\![0,z{-}1]\!]$ and $\ell\in [\![1,2^{\alpha}{-}1]\!]$. We have
\[
	L_{( k2^{z-\alpha-1},B)} = L_{( k2^{z-\alpha - 1} + k2^{z - \alpha}\ell ,B_\ell)}
\]
in $\Pi \left(\A_{m,2^p} \times \A_{\T,2^p} \right)$.
\end{proposition}

\begin{proof}
Let $v\in A_{2^p}^*$, $n=|v|$, $e=\val_{2^p} (v)$, $d= \frac{2^{pn}k2^{z-\alpha {-}1} + e}{m}$ and $d'=\frac{2^{pn}(k2^{z - \alpha - 1} + k2^{z - \alpha}\ell)+e}{m}$. We have to show that we have $d\in[\![0,2^{pn}{-}1]\!]$ and $B_d=T$ if and only if $d'\in[\![0,2^{pn}{-}1]\!]$ and $(B_\ell)_{d'}=T$. 

Using that $k\ge 1$, $e<2^{pn}$ and $\alpha<z$, we get
\begin{equation}
\label{eq:alpha}
	0<d 	= \frac{k2^{pn+z - \alpha - 1} + e}{m}
		< 2^{pn-\alpha{-}1}+\frac{2^{pn-z}}{k}
		\le 2^{pn-\alpha{-}1}+2^{pn-z}
		\le 2^{pn-\alpha}.
\end{equation}
Since $d'=d + \frac{k2^{pn+z-\alpha}\ell}{m}=d+2^{pn-\alpha}\ell$, we obtain that if both $d$ and $d'$ are integers then
\[
	\rep_2(d')=\rep_2(\ell)0^{pn-\alpha-|\rep_2(d)|}\rep_2(d),
\]
and hence $B_d=(B_\ell)_{d'}$.

Now, suppose that $d\in[\![0,2^{pn}{-}1]\!]$ and $B_d=T$. Then, we get from~\eqref{eq:alpha} that $pn > \alpha$ and $d'=d+2^{pn-\alpha}\ell$ is a nonnegative integer. Moreover, $d'<2^{pn-\alpha}(\ell+1)\le 2^{pn}$. Therefore $d'\in[\![0,2^{pn}{-}1]\!]$ and $(B_\ell)_{d'}=B_d=T$.

Conversely, suppose that $d'\in[\![0,2^{pn}{-}1]\!]$ and $(B_\ell)_{d'}=T$. In particular, we have that $d' \in \T\iff \ell\notin\T$. From~\eqref{eq:alpha}, we know that $0\le d<2^{pn}$. We claim that $pn > \alpha$. Proceed by contradiction and suppose that $pn \le \alpha$. Let $q =\DIV(\ell,2^{\alpha - pn})$. Then, on the one hand, we have 
\[
	d' 	=d+\frac{\ell}{2^{\alpha-pn}} \ge q.
\]
On the other hand, since $\ell\le (q+1)2^{\alpha - pn}{-}1$, $e<2^{pn}$, $k\ge 1$ and $\alpha<z$, we successively obtain that
\begin{align*}
	md' 	&< 2^{pn} ( k2^{z - \alpha {-}1} + k2^{z - \alpha}( (q+1)2^{\alpha - pn}{-}1)) +  2^{pn} \\
			&= m(q+1) + 2^{pn} (k2^{z - \alpha {-}1}  - k2^{z-\alpha} +  1) \\
		&= m(q+1) + 2^{pn} (1-k2^{z - \alpha {-}1}) \\
		&\le  m(q+1).
\end{align*}
We obtain that $q\le d'<(q+1)$, hence $d'=q$ and $\ell=2^{\alpha-pn}d'$, contradicting that $d' \in \T\iff \ell\notin\T$. Thus, we have that $pn > \alpha$ and $d=d'-2^{pn-\alpha}\ell$ is an integer. It follows that $d\in[\![0,2^{pn}{-}1]\!]$ and $B_d=T$.
\end{proof}

\begin{corollary}\label{cor:calpha}
For all $\alpha \in [\![0,z{-}1]\!]$, all states of the pre-class $C_{\alpha}$ are indistinguishable in $\Pi \left( \A_{m,2^p} \times \A_{\T,2^p} \right)$.
\end{corollary}

\begin{proposition}
\label{prop:zp-2}
Suppose that $z\ge 1$ and let $\beta \in[\![0,\big\lceil \frac zp \big\rceil{-}2]\!]$ and $c \in [\![1,p{-}1]\!]$. Then
\[
	L_{( k2^{z-\beta p-1},B)} = L_{( k2^{z-(\beta p+c)-1},B)}
\]
in $\Pi \left( \A_{m,2^p} \times \A_{\T,2^p} \right)$.
\end{proposition}

\begin{proof}
Let $v\in A_{2^p}^*$, $n=|v|$, $e=\val_{2^p} (v)$, $d=\frac{2^{pn}k2^{z-\beta p-1}+e}{m}$ and $d'=\frac{2^{pn}k2^{z-(\beta p+c)-1}+e}{m}$. We have to show that we have $d\in[\![0,2^{pn}{-}1]\!]$ and $B_d=T$ if and only if $d'\in[\![0,2^{pn}{-}1]\!]$ and $B_{d'}=T$. 

We have $d=2^{pn-\beta p{-}1}+\frac em$ and $d'=2^{pn-\beta p-c{-}1}+\frac em$. Since $z > (\beta+1)p$ and $p\ge c+1$, we have
\[
	\frac em<2^{pn-z}< 2^{pn-(\beta+1)p} \le 2^{pn-\beta p-c-1}.
\]
Thus, if both $d$ and $d'$ are integers and if $m$ divides $e$ then we obtain that
\[
	\rep_2(d)=10^{pn-\beta p{-}1-|\rep_2\left(\frac em\right)|}\rep_2\left(\frac em\right)
\]
and
\[
	\rep_2(d')=10^{pn-\beta p-c-1-|\rep_2\left(\frac em\right)|}\rep_2\left(\frac em\right).
\]
In this case, we have that $d\in\T\iff d'\in\T$, hence $B_d=B_{d'}$.

Now, suppose that $d\in[\![0,2^{pn}{-}1]\!]$ and $B_d=T$. Since $k\ge 1$ and $d =d'+2^{pn-\beta p{-}1}(1-2^{-c})$, we obtain that $0<d'<d<2^{pn}$. We claim that $pn\ge\beta p+c+1$. Proceed by contradiction and suppose that $pn<\beta p+c+1$. Then, since $c+1\le p$ and $\beta\le \big\lceil \frac zp \big\rceil -2$, we obtain that $pn\le\beta p<z-p$. Therefore, we have 
\[
	d=2^{pn-\beta p-1}+\frac em<\frac 12+\frac{2^{pn-z}}{k}<1
\] 
contradicting that $d$ is a positive integer. Thus $pn\ge\beta p+c+1$, and hence both $d'$ and $\frac em$ are integers. Therefore, we obtain that $d'\in[\![0,2^{pn}{-}1]\!]$ and $B_{d'}=T$.

Conversely, suppose that $d'\in[\![0,2^{pn}{-}1]\!]$ and $B_{d'}=T$. Using that $z\ge 1$, we obtain
\[
	0\le d = 2^{pn - \beta p -1} + \frac em < 2^{pn - \beta p-1} + 2^{pn-z} \le 2^{pn}.
\]
We claim that $pn\ge\beta p+c+1$. Proceed by contradiction and suppose that $pn<\beta p+c+1$. Since $c+1\le p$, we obtain that $n\le \beta$ and
\[
	d'=2^{pn - \beta p-c-1} + \frac em
	 < 2^{ -c - 1}+2^{pn-z} 
	 \le \frac 12+2^{\beta p-z}
	 < \frac 12+2^{-p}<1
\]
contradicting that $d'$ is a positive integer. Thus $d=d'-2^{pn-\beta p{-}1}(1-2^{-c})$ is an integer, and consequently, so is $\frac em$. Therefore, we obtain that $d\in[\![0,2^{pn}{-}1]\!]$ and $B_d=T$.
\end{proof}

\begin{corollary}\label{cor:gammabeta}
For all $\beta \in[\![0,\big\lceil \frac zp \big\rceil{-}2]\!]$, all states of the class $\Gamma_\beta$ are indistinguishable in $\Pi \left( \A_{m,2^p} \times \A_{\T,2^p} \right)$.
\end{corollary}

\begin{proposition}
Suppose that $z\ge 1$ and let $\beta=\big\lceil \frac zp \big\rceil{-}1$ and $c \in [\![1,z{-}\beta p{-}1]\!]$. We have
\[
	L_{( k2^{z-\beta p-1},B)} 
	= L_{( k2^{z-(\beta p+c)-1},B)}
\]
in $\Pi \left( \A_{m,2^p} \times \A_{\T,2^p} \right)$.
\end{proposition}

\begin{proof}
The proof is a straightforward adaptation of that of Proposition~\ref{prop:zp-2}.
\end{proof}

\begin{corollary}\label{cor:gammapascomplet}
In $\Pi \left( \A_{m,2^p} \times \A_{\T,2^p} \right)$, all states of $\Gamma_{\lceil \frac zp\rceil {-}1}$ are indistinguishable.
\end{corollary}

\subsection{States of different classes are distinguishable}
\label{sec:reduction2}

In this section, we show that, in the projected automaton $\Pi \left( \A_{m,2^p} \times \A_{\T,2^p} \right)$, states from different classes $[(j,X)]$ or $\Gamma_\beta$ are pairwise distinguishable, that is, for any two such states, there exists a word which is accepted from exactly one of them. 

First of all, note that the state $(0,T)$ is distinguished from all other states since it is the only final state: the empty word $\varepsilon$ is accepted from $(0,T)$ but not from any other state.

\begin{proposition}
Let $\beta \in [\![0,\big\lceil \frac zp \big\rceil {-}1]\!]$. In $\Pi \left( \A_{m,2^p} \times \A_{\T,2^p} \right)$, the word $0^{\beta +1}$ is accepted from all states of $\Gamma_{\beta}$.
\end{proposition}

\begin{proof}
From Corollaries~\ref{cor:gammabeta} and~\ref{cor:gammapascomplet}, it suffices to show that $0^{\beta +1}$ is accepted from the state $( k2^{z-\beta p-1},B)$. Let 
\[
	d=\frac{2^{p(\beta +1)}k2^{z-\beta p-1}}{m}.
\]
We have to show that $d\in\{0,\ldots,2^{p(\beta+1)}\}\setminus\T$. It is immediate since $d=2^{p -1}$. 
\end{proof}

\begin{proposition}
Let $\beta,\gamma \in[\![0,\big\lceil \frac zp \big\rceil {-}1]\!]$ such that $\gamma > \beta$. In $\Pi \left( \A_{m,2^p} \times \A_{\T,2^p} \right)$, the word $0^{\beta +1}$ is not accepted from any state of $\Gamma_{\gamma}$.
\end{proposition}

\begin{proof}
From Corollaries~\ref{cor:gammabeta} and~\ref{cor:gammapascomplet}, it suffices to show that $0^{\beta +1}$ is not accepted from the state $( k2^{z - \gamma p {-}1},B)$. Suppose  to the contrary that $0^{\beta +1}$ is accepted from $( k2^{z - \gamma p {-}1},B)$. Then
\[
	\frac{2^{p(\beta +1)}k2^{z - \gamma p -1}}{m} = 2^{p(\beta -\gamma +1)-1}
\]
must be an integer, and hence $p(\beta -\gamma +1)\ge 1$, contradicting that $\gamma > \beta$. The conclusion follows.
\end{proof}

\begin{proposition}
Let $(j,X) \in([\![1,k{-}1]\!]\times \{T,B\})\cup\{(0,B)\}$ and $\beta \in[\![0,\big\lceil \frac zp\big\rceil {-}1]\!]$. In $\Pi \left( \A_{m,2^p} \times \A_{\T,2^p} \right)$, the word $0^{\beta +1}$ is not accepted from any state of $[(j,X)]$.
\end{proposition}

\begin{proof}
Since there is a loop labeled by $0$ on the state $(0,T)$ and in view of Corollary~\ref{cor:0B}, it suffices to show that the word $0^{\lceil z/p\rceil}$ is not accepted from the state $(j,X)$. If $0^{\lceil z/p\rceil}$ were accepted from the state $(j,X)$, then we would get that
\[
	d=\frac{2^{p\left\lceil\frac zp\right\rceil}j}{m}=\frac{2^{p\left\lceil\frac zp\right\rceil-z}j}{k}
\]
is an integer and that $X_d=T$. If $j\ne 0$, then $d$ cannot be an integer since $k$ is odd and $0<j<k$. If $j=0$, then we get that $d$ must belong to $\T$, which is not possible either since in this case we have $d=0$. Hence the conclusion.
\end{proof}

\begin{proposition}
Suppose that $k>1$ and let $(j,X),(j',X')\in([\![1,k{-}1]\!]\times \{T,B\})\cup\{(0,B)\}$ be distinct. In $\Pi \left(\A_{m,2^p} \times \A_{\T,2^p}  \right)$, the states $(j,X)$ and $(j',X')$ are distinguishable.
\end{proposition}

\begin{proof}
First, suppose that $j=j'$. Then $X\ne X'$ by hypothesis and the states $(j,X)$ and $(j,X')$ are disjoint by Lemma~\ref{lem:projdisj}. Since $\Pi \left(\A_{m,2^p} \times \A_{\T,2^p}  \right)$ is coaccessible by Proposition~\ref{prop:m-odd}, we obtain that the states $(j,X)$ and $(j,X')$ are distinguishable.

Now suppose that $j\ne j'$. By Proposition~\ref{prop:jj'}, the word $w_j$ is accepted from $j$ in the automaton $\Pi(\A_{m,2^p})$ but is not accepted from $j'$. Then, there exists a word $u$ of length $|w_j|$ such that $(u,w_j)$ is accepted from $j$ in the automaton $\A_{m,2^p}$ but is not accepted from $j'$. Then, this word $(u,w_j)$ is accepted either from $(j,T)$ or from $(j,B)$ in the automaton $\A_{m,2^p}\times \A_{\T,2^p}$ but is not accepted neither from $(j',T)$ nor from $(j',B)$. Now, two cases are possible. 

First, suppose that $(u,w_j)$ is accepted from $(j,X)$ in $\A_{m,2^p}\times \A_{\T,2^p}$. Then, in the projection $\Pi \left(\A_{m,2^p} \times \A_{\T,2^p}  \right)$, the word $w_j$ is accepted from $(j,X)$ but not from $(j',X')$. Thus, the word $w_j$ distinguishes the states $(j,X)$ and $(j',X')$. 

Second, suppose that $(u,w_j)$ is accepted from $(j,\overline{X})$ in $\A_{m,2^p}\times \A_{\T,2^p}$. Then there is a path labeled by $(u,w_j)$ from $(j,X)$ to $(0,B)$ in $\A_{m,2^p}\times \A_{\T,2^p}$. By Corollary~\ref{cor:0Baccessible}, in $\A_{m,2^p}\times \A_{\T,2^p}$, the word $\rep_{2^p}(1,m)$ is accepted from $(0,B)$, and hence the word $(u,w_j)\rep_{2^p}(1,m)=(u0^{|\rep_{2^p}(m)|-1}1,w_j\rep_{2^p}(m))$ is accepted from $(j,X)$. Therefore the word $w_j\rep_{2^p}(m)$ is accepted from the state $(j,X)$ in $\Pi \left(\A_{m,2^p} \times \A_{\T,2^p}  \right)$. Besides, the word $w_j\rep_{2^p}(m)$ cannot be accepted from $(j',X')$ in $\Pi \left(\A_{m,2^p} \times \A_{\T,2^p}  \right)$ for otherwise it would also be accepted from $j'$ in $\Pi \left(\A_{m,2^p} \right)$, which is impossible by Proposition~\ref{prop:jj'-bis}. Thus, the word $w_j\rep_{2^p}(m)$ distinguishes the states $(j,X)$ and $(j',X')$. 
\end{proof}

\begin{corollary}
In the automaton $\Pi \left( \A_{m,2^p} \times \A_{\T,2^p} \right)$, two states belonging to different classes are distinguished.
\end{corollary}

\subsection{The minimal automaton of $\val_{2^p}^{-1}(m \T)$.}
We are ready to construct the minimal automaton of $\val_{2^p}^{-1}(m \T)$. Since the states of $\Pi \left( \A_{m,2^p} \times \A_{\T,2^p} \right)$ that belong to the same class $[(j,X)]$ or $\Gamma_\beta$ are indistinguishable, they can be glued together in order to define a new automaton $\mathcal{M}_{m,\T,2^p}$ that still accepts the same language. Formally, the alphabet of  $\mathcal{M}_{m,\T,2^p}$ is $A_{2^p}$. Its states are the classes $[(j,X)]$ for $(j,X)\in[\![0,k{-}1]\!]\times \{T,B\}$ and the classes $\Gamma_\beta$ for $\beta\in[\![0,\lceil \frac zp\rceil{-}1]\!]$. The class $[(0,T)]$ is the initial state and the only final state. The transitions of $\mathcal{M}_{m,\T,2^p}$ are defined as follows: there is a transition labeled by a letter $a$ in $A_{2^p}$ from a class $J_1$ to a class $J_2$ if and only if there exists $j_1 \in J_1$ and $j_2 \in J_2$ such that, in the automaton $\Pi \left( \A_{m,2^p}\times \A_{\T,2^p} \right)$, there is a transition labeled by $a$ from the state $j_1$ to the state $j_2$. 

\begin{example}
In Figure~\ref{fig:classe-6T}, the classes of $\Pi \left( \A_{6,4}\times \A_{\T,4} \right)$ are colored in white, blue, grey, yellow, fushia, orange and purple.
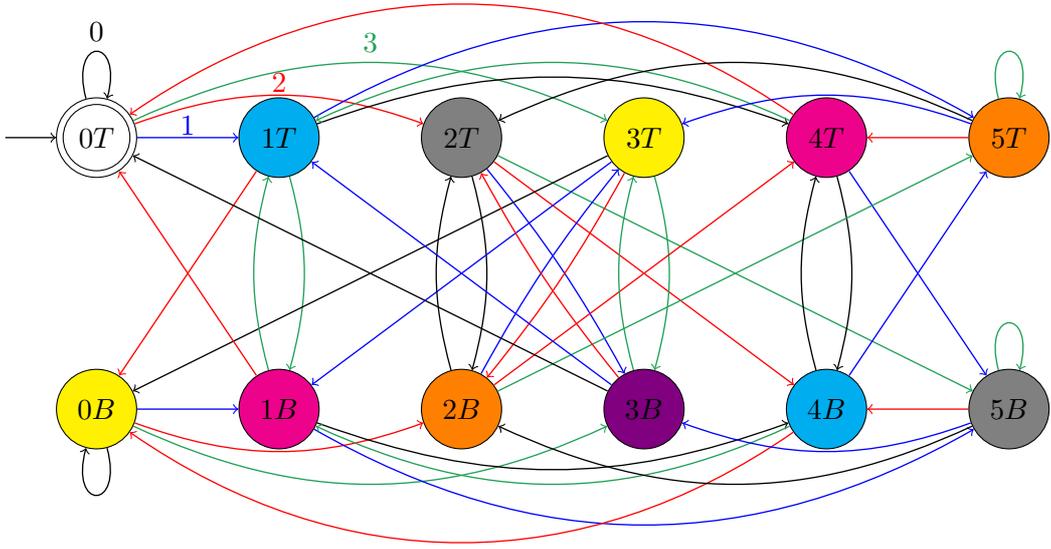
\begin{figure}[htb]
\centering
\begin{tikzpicture}[scale=0.8]
\tikzstyle{every node}=[shape=circle, fill=none, draw=black,
minimum size=30pt, inner sep=2pt]
\node(0T) at (0,0) {$0T$};
\node[fill=cyan](1T) at (3,0) {$1T$};
\node[fill=gray](2T) at (6,0) {$2T$};
\node[fill=yellow](3T) at (9,0) {$3T$};
\node[fill=magenta](4T) at (12,0) {$4T$};
\node[fill=orange](5T) at (15,0) {$5T$};
\node[fill=yellow](0B) at (0,-4.5) {$0B$};
\node[fill=magenta](1B) at (3,-4.5) {$1B$};
\node[fill=orange](2B) at (6,-4.5) {$2B$};
\node[fill=violet](3B) at (9,-4.5) {$3B$};
\node[fill=cyan](4B) at (12,-4.5) {$4B$};
\node[fill=gray](5B) at (15,-4.5) {$5B$};
\tikzstyle{every node}=[shape=circle, fill=none, draw=black,minimum size=25pt, inner sep=2pt]
\node at (0,0) {};

\tikzstyle{etiquettedebut}=[very near start,rectangle,fill=black!20]
\tikzstyle{etiquettemilieu}=[midway,rectangle,fill=black!20]
\tikzstyle{every path}=[color=black, line width=0.5 pt]
\tikzstyle{every node}=[shape=circle, minimum size=5pt, inner sep=2pt]
\draw [->] (-1.5,0) to node {} (0T); 
\draw [->] (0T) to [loop above] node [] {$0$} (0T);
\draw [green,->] (5T) to [loop above] node [] {} (5T);
\draw [->] (0B) to [loop below] node [] {} (0B);
\draw [green,->] (5B) to [loop above] node [] {} (5B);
\draw [blue,->] (0T) to [] node [above=-0.1] {$1$} (1T);
\draw [red,->] (0T) to [bend left=20] node [above=-0.1] {$2$} (2T);
\draw [green,->] (0T) to [bend left=25] node [above] {$3$} (3T);
\draw [->] (1T) to [bend left=20] node [] {} (4T);
\draw [blue,->] (1T) to [bend left=30] node [] {} (5T);
\draw [red,->] (1T) to [] node [] {} (0B);
\draw [green,->] (1T) to [bend left=15] node [] {} (1B);
\draw [->] (2T) to [bend left=15] node [] {} (2B);
\draw [blue,->] (2T) to [bend left=5] node [] {} (3B);
\draw [red,->] (2T) to [] node [] {} (4B);
\draw [green,->] (2T) to [] node [] {} (5B);
\draw [->] (3T) to [] node [] {} (0B);
\draw [blue,->] (3T) to [] node [] {} (1B);
\draw [red,->] (3T) to [bend left=5] node [] {} (2B);
\draw [green,->] (3T) to [bend left=15] node [] {} (3B);
\draw [->] (4T) to [bend left=15] node [] {} (4B);
\draw [blue,->] (4T) to [] node [] {} (5B);
\draw [red,->] (4T) to [bend right=35] node [] {} (0T);
\draw [green,->] (4T) to [bend right=25] node [] {} (1T);
\draw [->] (5T) to [bend right=25] node [] {} (2T);
\draw [blue,->] (5T) to [bend right=20] node [] {} (3T);
\draw [red,->] (5T) to [] node [] {} (4T);

\draw [blue,->] (0B) to [] node [] {} (1B);
\draw [red,->] (0B) to [bend right=20] node [] {} (2B);
\draw [green,->] (0B) to [bend right=25] node [] {} (3B);
\draw [->] (1B) to [bend right=20] node [] {} (4B);
\draw [blue,->] (1B) to [bend right=30] node [] {} (5B);
\draw [red,->] (1B) to [] node [] {} (0T);
\draw [green,->] (1B) to [bend left=15] node [] {} (1T);
\draw [->] (2B) to [bend left=15] node [] {} (2T);
\draw [blue,->] (2B) to [bend left=5] node [] {} (3T);
\draw [red,->] (2B) to [] node [] {} (4T);
\draw [green,->] (2B) to [] node [] {} (5T);
\draw [->] (3B) to [] node [] {} (0T);
\draw [blue,->] (3B) to [] node [] {} (1T);
\draw [red,->] (3B) to [bend left=5] node [] {} (2T);
\draw [green,->] (3B) to [bend left=15] node [] {} (3T);
\draw [->] (4B) to [bend left=15] node [] {} (4T);
\draw [blue,->] (4B) to [] node [] {} (5T);
\draw [red,->] (4B) to [bend left=35] node [] {} (0B);
\draw [green,->] (4B) to [bend left=25] node [] {} (1B);
\draw [->] (5B) to [bend left=25] node [] {} (2B);
\draw [blue,->] (5B) to [bend left=20] node [] {} (3B);
\draw [red,->] (5B) to [] node [] {} (4B);
\end{tikzpicture}
\caption{The classes of the automaton of $\Pi \left( \A_{6,4}\times \A_{\T,4} \right)$.}
\label{fig:classe-6T}
\end{figure}
Figure~\ref{fig:min-aut-6T} depicts the minimal automaton $\mathcal{M}_{6,\T,4}$ of $\val_4^{-1}(6\T)$, where states corresponding to the same color are glued together to form a single state. 
\begin{figure}[h!]
\centering
\begin{tikzpicture}
[scale=0.8]
\tikzstyle{every node}=[shape=circle, fill=none, draw=black,
minimum size=30pt, inner sep=2pt]
\node(blanc) at (0,0) {};
\node[fill=cyan](bleu) at (4,0) {};
\node[fill=gray](gris) at (8,0) {};
\node[fill=yellow](jaune) at (12,0) {};
\node[fill=magenta](rose) at (2,-4.5) {};
\node[fill=orange](orange) at (6,-4.5) {};
\node[fill=violet](mauve) at (10,-4.5) {};
\tikzstyle{every node}=[shape=circle, fill=none, draw=black,
minimum size=25pt, inner sep=2pt]
\node at (0,0) {};

\tikzstyle{etiquettedebut}=[very near start,rectangle,fill=black!20]
\tikzstyle{etiquettemilieu}=[midway,rectangle,fill=black!20]
\tikzstyle{every path}=[color=black, line width=0.5 pt]
\tikzstyle{every node}=[shape=circle, minimum size=5pt, inner sep=2pt]

\draw [->] (-1.5,0) to node {} (blanc);
\draw [black,->] (blanc) to [loop above] node [] {$0$} (blanc);
\draw [blue,->] (blanc) to [] node [above] {$1$} (bleu);
\draw [red,->] (blanc) to [bend left=35] node [below] {$2$} (gris);
\draw [green,->] (blanc) to [bend left=40] node [below] {$3$} (jaune);
\draw [black,->] (bleu) to [bend left=15] node [] {} (rose);
\draw [blue,->] (bleu) to [] node [] {} (orange);
\draw [red,->] (bleu) to [bend left=35] node [] {} (jaune);
\draw [green,->] (bleu) to [bend left=5] node [] {} (rose);
\draw [black,->] (gris) to [bend left=10] node [] {} (orange);
\draw [blue,->] (gris) to [bend left=10] node [] {} (mauve);
\draw [red,->] (gris) to [] node [] {} (bleu);
\draw [green,->] (gris) to [loop above] node [] {} (gris);
\draw [black,->] (jaune) to [loop above] node [] {} (jaune);
\draw [blue,->] (jaune) to [] node [] {} (rose);
\draw [red,->] (jaune) to [bend left=5] node [] {} (orange);
\draw [green,->] (jaune) to [bend left=10] node [] {} (mauve);
\draw [black,->] (rose) to [bend left=15] node [] {} (bleu);
\draw [blue,->] (rose) to [] node [] {} (gris);
\draw [red,->] (rose) to [] node [] {} (blanc);
\draw [green,->] (rose) to [bend left=5] node [] {} (bleu);
\draw [black,->] (orange) to [bend left=10] node [] {} (gris);
\draw [blue,->] (orange) to [bend left=5] node [] {} (jaune);
\draw [red,->] (orange) to [] node [] {} (rose);
\draw [green,->] (orange) to [loop below] node [] {} (orange);
\draw [black,->] (mauve) to [] node [] {} (blanc);
\draw [blue,->] (mauve) to [] node [] {} (bleu);
\draw [red,->] (mauve) to [bend left=10] node [] {} (gris);
\draw [green,->] (mauve) to [bend left=10] node [] {} (jaune);
\end{tikzpicture}
\caption{The minimal automaton $\mathcal{M}_{6,\T,4}$ of $\val^{-1}_4(6\T)$.}
\label{fig:min-aut-6T}
\end{figure}
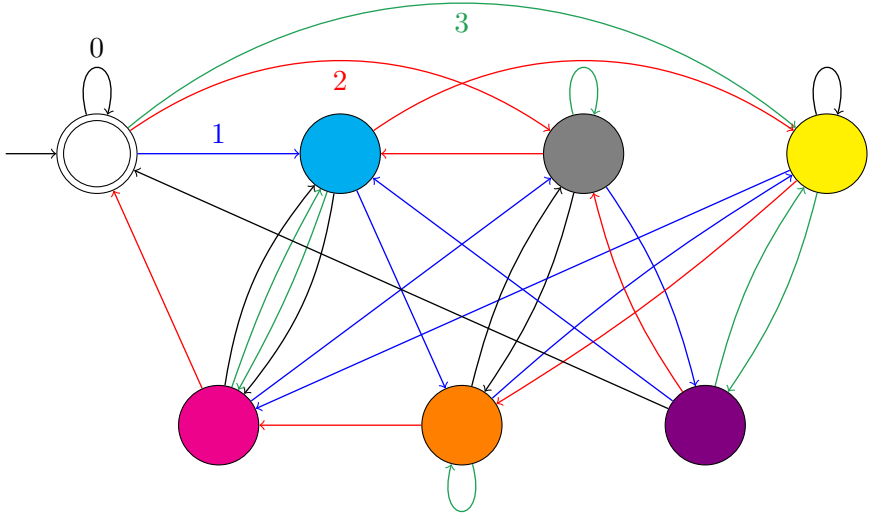
\end{example}

\begin{theorem}
\label{thm:main2}
Let $p$ and $m$ be positive integers. The automaton $\mathcal{M}_{m,\T,2^p}$ is the minimal automaton of the language $\val_{2^p}^{-1}(m \T)$. 
\end{theorem}

\begin{proof}
By construction, the language accepted by $\mathcal{M}_{m,\T,2^p}$ is $\val_{2^p}^{-1}(m \T)$. In order to see that $\mathcal{M}_{m,\T,2^p}$ is minimal, it suffices to prove that it is complete, reduced and accessible. The fact that $\mathcal{M}_{m,\T,2^p}$ is reduced follows from the results of Sections~\ref{sec:reduction1} and~\ref{sec:reduction2}. We know from Proposition~\ref{prop:m-odd} that the automaton $\Pi \left( \A_{m,2^p} \times \A_{\T,2^p} \right)$ is complete and accessible, which in turn implies that $\mathcal{M}_{m,\T,2^p}$ is complete and accessible as well. 
\end{proof}

Note that Proposition~\ref{prop:m-odd} and Theorem~\ref{thm:main2} are consistent in the case where $m$ is odd, i.e.\ where $z=0$.

We are now ready to prove Theorem~\ref{thm:main}.

\begin{proof}[Proof of Theorem~\ref{thm:main}]
In view of Theorem~\ref{thm:main2}, it suffices to count the number of states of $\mathcal{M}_{m,\T,2^p}$. By definition, it has $2(k{-}1)+2=2k$ states of the form $[(j,X)]$ and $\lceil \frac{z}{p}\rceil$ states of the form $\Gamma_\beta$.
\end{proof}

\begin{example}
The minimal automaton of the language $\rep_4 (6 \T)$  has $7$ states; see Figure~\ref{fig:min-aut-6T}. We can indeed compute that $2\cdot 3+\lceil \frac{1}{2} \rceil = 7$.
\end{example}

\bibliographystyle{abbrv}
\bibliography{TMmultiples}
\label{sec:biblio}

\end{document}